\documentclass[11pt]{amsart}

\usepackage{a4wide}

\usepackage{amscd}
\usepackage{amssymb}
\usepackage{amsthm}
\usepackage{amsfonts}
\usepackage{amsmath}
\usepackage{latexsym}
\usepackage{verbatim}
\usepackage{comment}
\usepackage{epsfig}
\usepackage{graphicx}
\usepackage{epstopdf}
\usepackage{stmaryrd}
\usepackage{booktabs}
\usepackage{float}

\makeatletter
\renewcommand\paragraph{\@startsection{paragraph}{4}{\z@}%
  {.8\baselineskip}   
  {.5\baselineskip}   
  {\normalfont\normalsize\bfseries}} 
\makeatother

\newtheorem{algorithm}{Algorithm}
\newtheorem{prop}{Proposition}
\newtheorem{theorem}{Theorem}
\newtheorem{lemma}{Lemma}
\newtheorem{assumption}{Assumption}
\newtheorem{corollary}[theorem]{Corollary}
\newtheorem{definition}{Definition}

\theoremstyle{remark}
\newtheorem{remark}{Remark}

\usepackage{graphicx}
\usepackage{subcaption}
\usepackage{float}

\newcommand{\R}{\ensuremath{\mathbb{R}}}
\newcommand{\E}{\ensuremath{\mathbb{E}}}

\newcommand{\C}{\ensuremath{\mathbb{C}}}
\def\e{{\mathrm{e}}}

\renewcommand{\Re}{\operatorname{Re}}

\allowdisplaybreaks

\usepackage{color}
\definecolor{darkgreen}{rgb}{0, .5, 0}
\definecolor{darkred}{rgb}{.5, 0, 0}

\title[Pricing Derivatives under Self-Exciting Dynamics]{Pricing Derivatives under Self-Exciting Dynamics: A Finite-Difference and Transform Approach}

\author{Aqib Ahmed and Heidar Eyjolfsson}

\thanks{The authors are grateful for funding from the Sustainability Institute and Forum at Reykjavík University, project number 224089 supported by Landsnet; and from the Energy Research fund of Landsvirkjun. The grants contributed to the funding of the project \emph{Climate change, finance, and sustainability} and enabled us to carry out the research that is presented in this article.
A. Ahmed would like to personally thank Imane Fakir for her invaluable support and many insightful discussions related to this article.}

\begin{document}

\begin{abstract}
We consider the pricing of derivatives written on accumulated marks, such as weather derivatives or aggregate loss claims, using a self-exciting marked point process. The jump intensity mean-reverts between events and increases at jump times by an amount proportional to the mark. The resulting state process, where the variable $U_t$ accumulates jump magnitudes, is a piecewise deterministic Markov process (PDMP). We derive the discounted pricing equation as a backward partial integro-differential equation (PIDE) in two spatial dimensions. To overcome the dimensionality, we propose an exponential (Laplace/Fourier) transform in the accumulated mark variable, which diagonalizes the translation operator and reduces the pricing problem to a family of one-dimensional PIDEs in the intensity variable along a Bromwich contour. For Gamma-mixture mark laws (under actuarial or Esscher-tilted measures), the nonlocal jump term is efficiently approximated by generalized Gauss--Laguerre quadrature. We solve the reduced PIDEs backward in time using a monotone IMEX finite difference scheme (implicit upwind drift and discounting, explicit jump operator) and recover option prices via numerical inversion. We provide a rigorous, term-by-term global error bound covering time and space discretization, quadrature, interpolation, and boundary effects, supported by numerical experiments and Monte Carlo benchmarks.
\end{abstract}

\maketitle

\section{Introduction}

The pricing of derivatives written on accumulated environmental or physical variables—such as rainfall, heating/cooling degree days, or aggregate insurance losses—requires models capable of capturing the erratic and clustered nature of the underlying events. Unlike traditional financial assets, these underlying variables are non-tradable and exhibit distinct periods of high activity followed by relative calm. Such clustering (or contagion) implies that the occurrence of an event, or a sudden large-magnitude mark, increases the probability of subsequent events. This makes self-exciting point processes a natural and highly effective candidate for modelling the underlying dynamics in weather and catastrophe risk management.

To capture self-excitation, Hawkes processes \cite{Hawkes1971a, Hawkes1971b} and their variants have been extensively applied in finance (see e.g. Bacry et al. \cite{Bacry2013}, Aït-Sahalia et al. \cite{AitSahalia2015}, Callegaro et al. \cite{Callegaro22}, Ba\~nos et al. \cite{Banos24}) and insurance (see e.g. Dassios and Zhao \cite{DaZh17}, Errais et al. \cite{Errais10}, Swishchuk et al. \cite{Swishchuk21}). Recently, Eyjolfsson and Tjøstheim \cite{ET18, ET21} introduced a flexible class of piecewise deterministic Markov processes (PDMPs) in the sense of Davis \cite{Davis93}, where the event intensity is driven by a stochastic differential equation (SDE). In this framework, the intensity mean-reverts during calm periods but jumps proportionally to the mark size at event times. While statistically powerful for capturing the feedback loop between event magnitudes and arrival rates, pricing contingent claims within this PDMP framework remains computationally challenging.

By Dynkin's formula, the discounted value of a payoff depending on the terminal accumulated mark $U_T$ satisfies a backward partial integro-differential equation (PIDE). In our setting, the state space is two-dimensional: the intensity $\lambda$ and the accumulated mark $u$. The jump operator involves a translation $u \mapsto u+x$, rendering the PIDE non-local in both variables. Solving this two-dimensional PIDE directly via finite differences is computationally prohibitive and prone to severe numerical diffusion, especially for non-smooth payoffs such as capped calls.

To overcome this curse of dimensionality, we propose a spectral reduction technique inspired by the Fourier pricing methods of Carr and Madan \cite{CarrMadan1999} and the exponential-Lévy framework (see Cont and Tankov \cite{ContTankov2004}). By applying an exponential Laplace/Fourier (Bromwich) transform with respect to the accumulation variable $u$, we exploit the fact that exponential functions diagonalize the translation operator. The ansatz $V(t,\lambda,u) = e^{\eta u} F(t,\lambda,\eta)$ effectively factors out $u$, reducing the 2D pricing problem to a family of independent 1D PIDEs in $\lambda$, parameterized by the complex frequency $\eta = \delta + iy$. The option price is then recovered by integrating along the Bromwich contour, equivalently written as a one-dimensional real-valued integral over the Fourier variable $y$.

For the numerical resolution of the reduced 1D PIDEs, we develop a bespoke Implicit-Explicit (IMEX) finite difference scheme, adapting the philosophy of Cont and Voltchkova \cite{ContVoltchkova2005}. The local mean-reverting drift and discounting are treated implicitly using an unconditionally stable upwind scheme, while the non-local jump integral is treated explicitly. To evaluate the jump integral over the semi-infinite domain without introducing artificial truncation, we leverage the structure of the mark distribution. By modelling the marks as a Gamma mixture-of-experts (which conveniently remains a Gamma mixture under the Esscher transform), the integral is computed efficiently and exactly using generalized Gauss--Laguerre quadrature. Off-grid evaluations are handled via shape-preserving Piecewise Cubic Hermite Interpolation (PCHIP) or linear interpolation to guarantee monotonicity and prevent spurious oscillations.

The main contributions of this paper are fourfold:
\begin{itemize}
    \item We formulate a rigorous PDMP pricing framework for cumulative payoffs under self-exciting SDE-driven intensities, defining the extended generator and the corresponding PIDE under both actuarial and Esscher-tilted pricing measures.
    \item We prove a superposition/inversion theorem that validates the reduction of the 2D PIDE to a family of 1D PIDEs via a Bromwich contour integral.
    \item We design a robust IMEX--Gauss--Laguerre numerical scheme and prove its unconditional implicit stability (M-matrix property) and explicit $L^\infty$-Lipschitz bounds.
    \item We provide a comprehensive, term-by-term global convergence analysis, bounding the errors arising from time-stepping, spatial upwinding, quadrature, interpolation, and Bromwich truncation.
\end{itemize}

The remainder of the paper is organized as follows. Section \ref{sec:notations} defines the PDMP framework, the extended generator, the Esscher pricing measure, and the Gamma mixture-of-experts mark law. Section \ref{sec:framework} presents the spectral reduction via the Bromwich transform and proves the inversion theorem. Section \ref{sec:IMEX} details the discrete IMEX--Gauss--Laguerre scheme and establishes the global convergence theorem. Section \ref{sec:global-reconstruction-error} discusses the discrete Fourier reconstruction and its error bounds. Finally, Section \ref{sec:numerics} provides numerical experiments illustrating the accuracy of the algorithm and the economic impact of self-excitation, benchmarked against exact Monte Carlo simulations.

\section{Notation and Setting}\label{sec:notations}

We work on a filtered probability space $(\Omega,\mathcal F,\mathbb F,\mathbb P)$ satisfying the usual hypotheses.
Let $N(dt,dx)$ be an $\mathbb F$--adapted marked point process on $\R_+\times\R_+$, with associated counting process
$N_t:=N((0,t]\times\R_+)$ and jump marks $(X_k)_{k\ge1}\subset\R_+$.
Define the cumulative mark (compound) process
\[
U_t := \int_{(0,t]\times\R_+} x\,N(ds,dx)=\sum_{k=1}^{N_t}X_k.
\]
We consider the space-time state process $Z_t:=(\lambda_t,N_t,U_t)$, where $\lambda_t$ is the event intensity.
Between jumps, $\lambda_t$ follows an ODE, and at each jump time it increases proportionally to the mark,
so that $Z_t$ is a piecewise deterministic Markov process (PDMP) with jumps; see, e.g., the PDMP framework of Davis \cite{Davis93} and the marked point process setting of Br\'emaud \cite{Bremaud1981}.

Building on \cite{DaZh17,ET18,ET21}, we model the intensity as
\begin{equation}\label{eq:intensity-dyn}
d\lambda_t = \mu(\lambda_t)\,dt + \beta\,dU_t,
\end{equation}
where $\mu:\R_+\to\R$ is locally Lipschitz with $\mu(0)\ge0$ (in applications $\mu(\lambda)=\kappa(\bar\lambda-\lambda)$) and $\beta\ge0$ is the self-excitation parameter.
Assume that $N(dt,dx)$ admits an $\mathbb F$--predictable intensity process $(\lambda_t)_{t\ge0}$ in the sense of \cite{Bremaud1981}.
Conditionally  on the pre-jump intensity $\lambda_{t-}=\lambda$, the mark distribution is given by a Borel probability kernel
\[
\nu(\lambda,dx)\quad \text{on }\R_+,
\]
and the $\mathbb F$--compensator of $N(dt,dx)$ under $\mathbb P$ is
\[
\lambda_{t-}\,\nu(\lambda_{t-},dx)\,dt.
\]
When switching to a pricing measure, we denote the corresponding mark kernel by $\tilde\nu(\lambda,dx)$ (and keep the same notation $\lambda_t$ for the intensity under that measure when no ambiguity arises). Fix a compact computational domain $\Lambda:=[\lambda_{\min},\lambda_{\max}]$.
For a bounded function $g:\Lambda\to\C$, set
\[
\|g\|_{\infty,\Lambda}:=\sup_{\lambda\in\Lambda}|g(\lambda)|.
\]
All functions we employ are continuous on \(\Lambda\), hence the essential supremum equals the supremum.

For $\delta>0$, define the exponential moment functional (under the chosen pricing kernel $\tilde\nu$)
\[
C_\delta(\lambda):=\int_{0}^{\infty} e^{\delta x}\,\tilde\nu(\lambda,dx),
\qquad
C_\delta^*(\Lambda):=\sup_{\lambda\in\Lambda} C_\delta(\lambda).
\]
For two finite signed Borel measures $\mu,\nu$ on $\R_+$, write $\sigma=\mu-\nu$ and denote by
$|\sigma|$ its total variation measure (Jordan decomposition). For $\delta\ge0$ such that
$\int_0^\infty e^{\delta x}\,|\sigma|(dx)<\infty$, define the weighted total variation norm
\[
\|\mu-\nu\|_{TV,\delta}:=\int_0^\infty e^{\delta x}\,|\mu-\nu|(dx)
=\sup_{|g(x)|\le e^{\delta x}}\left|\int_0^\infty g(x)\,(\mu-\nu)(dx)\right|.
\]
(Equivalently, $\|\cdot\|_{TV,\delta}$ is the usual total variation norm applied to the signed measure with density $e^{\delta x}$. In particular, it is a norm on the linear space of signed measures with finite $\|\cdot\|_{TV,\delta}$.)

\paragraph{Clamped extension (off-grid evaluation).}
We evaluate off--grid queries with the clamped extension
\[
(E_{\rm clamp}F)(\lambda') :=
\begin{cases}
F(\lambda'), & \lambda'\in\Lambda,\\
F(\lambda_{\min}), & \lambda'<\lambda_{\min},\\
F(\lambda_{\max}), & \lambda'>\lambda_{\max}.
\end{cases}
\]
This extension is non--expansive in sup--norm:
$\|E_{\rm clamp}F-E_{\rm clamp}G\|_{\infty,[0,\infty)}=\|F-G\|_{\infty,\Lambda}$.
Whenever an operator requires values outside $\Lambda$ (e.g.\ $N_{\exp}$),
we use either the clamped extension $E_{\rm clamp}$ (default) or a short linear extrapolation at the boundary for Gauss--Laguerre off--domain nodes;
see Appendix~\ref{app:PCHIP-boundary}.

\medskip
\noindent Consider the extended generator of the process $t \mapsto (t,\lambda_t,N_t,U_t)$. It has the form
\begin{equation}\label{def:Gen}
    \mathcal{A}f(t,\lambda,n,u) = \frac{\partial f}{\partial t} + \mu(\lambda)\frac{\partial f}{\partial \lambda} + \lambda \int_0^\infty \big(f(t,\lambda +\beta x,n+1,u+x) - f(t,\lambda,n,u)\big)\nu(\lambda,dx).
\end{equation}

\begin{prop}\label{prop:Generator}
    If for any $t > 0$ and a given measurable function $f:\R^4 \to \R$, it holds that $f$ has continuous partial derivatives with respect to $t$, $\lambda$, $u$ the map 
    $$
    (t,\lambda,n,u) \mapsto \int_0^\infty f(t,\lambda+\beta x,n+1,u+x)\nu(\lambda,dx)
    $$
    is measurable, with $\mathbb E_{\eta}$ the expectation under $\mathbb P$ given $X_0=\eta$
    \begin{equation}\label{f:condition}
    \E_\eta\bigg [ \int_0^t\lambda_s\int_0^\infty\{f(s,\lambda_{s-}+\beta x,N_{s-}+1,U_{s-}+x) - f(s,\lambda_{s-},N_{s-},U_{s-})\}\nu(\lambda,dx)ds\bigg] < \infty
    \end{equation}
    then $f$ is in the domain of the extended generator of the process $t \mapsto (t,\lambda_t,N_t,U_t)$, i.e. $f \in \mathcal{D}(\mathcal{A})$, where $\mathcal{A}$ is given by \eqref{def:Gen}. 
\end{prop}
\begin{proof}
    The proof is analogous to the proof of Proposition 1 in Eyjolfsson and Tj\o stheim \cite{ET21}.
\end{proof}


\section{Framework and partial integro-differential equations}\label{sec:framework}

Notice that our class of self-exciting processes is a piecewise deterministic Markov process (PDMP) in the sense of Davis \cite{Davis93}. We need the following assumption.
\begin{assumption}\label{ass:pdmp}
\noindent We work on $E=\R_+\times\R_+$ with state $z=(\lambda,u)$ and process $Z_t=(\lambda_t,U_t)$. We will apply the extended generator \eqref{def:Gen} to test functions independent of $N_t$.
\begin{itemize}
    \item[(i)]\textbf{Flow between jumps}. Between jumps, $\lambda$ follows the ODE $\dot\lambda=\mu(\lambda)$, with $\mu$ locally Lipschitz and of at most linear growth; $u$ is constant between jumps.
    \item[(ii)]\textbf{Jump intensity and kernel}. At state $(\lambda,u)$, the jump rate is $\lambda\ge 0$. Conditional on a jump, draw $x>0$ with kernel $\nu(\lambda,dx)$ (measurable in $\lambda$), and update $(\lambda,u)\mapsto(\lambda+\beta x,\,u+x)$.
    \item[(iii)]\noindent\textbf{Local integrability}. On any finite horizon $[t,T]$, $\E\!\big[\int_t^T \lambda_s\,ds\big]<\infty$ (finite expected number of jumps).
    \item[(iv)]\textbf{ Integrable jump term}. For any $F=F(t,\lambda,u)$ used with the generator, the jump term is locally integrable: $\lambda\!\int_0^\infty \!\big|F(t,\lambda+\beta x,u+x)-F(t,\lambda,u)\big|\,\nu(\lambda,dx)<\infty.$
    \item[(v)]\textbf{Domain of the generator}. Functions, $V$, to which the generator is applied are differentiable along the flow (a.e.) and satisfy (iv); in particular $V \in \mathcal D(\mathcal A)$.
\end{itemize}
\end{assumption}

\begin{prop}[Discounted value function and Backward PIDE]\label{prop:PIDE_discount}
Fix $r\ge 0$ and suppose $f\in C_b(\R_+)$ (bounded and continuous) or satisfies a growth condition such that $\E[e^{-r(T-t)}|f(U_T)|]<\infty$ and $f$ is bounded on compact sets. Then under Assumption ~\ref{ass:pdmp}, for $t<T$ the function 
\begin{equation}\label{eq:V_discount}
V(t,\lambda,u):=\E\!\left[e^{-r(T-t)}\,f(U_T)\ \big|\ \lambda_t=\lambda,\ U_t=u\right],
\qquad V(T,\lambda,u)=f(u).
\end{equation}

\noindent is the unique absolutely continuous solution of the PIDE 
\begin{equation}\label{eq:PIDE_discount}
\boxed{\;
(\mathcal{A}-r)V = 0}
\end{equation}
where $\mathcal{A}$ is given by \eqref{def:Gen} with terminal condition \eqref{eq:V_discount}.

\end{prop}

\begin{proof}
This follows from Dynkin's formula for PDMPs (see Davis~\cite[Sec.~32]{Davis93}). Applying the extended generator to the discounted process $s \mapsto e^{-r(s-t)}V(s, \lambda_s, U_s)$ and equating the martingale part to zero yields $(\mathcal A-r)V=0$. Uniqueness holds in $C_b$.
\end{proof}


\medskip
\subsection{Pricing measure}
\noindent Because the market is incomplete, there is no canonical risk–neutral measure.
We consider two pricing conventions that both lead to the backward template
$\partial_t V+\mu\,\partial_\lambda V+\mathsf J[V]-rV=0$,
with different jump operators $\mathsf J$.
\smallskip
\begin{enumerate}
  \item \textbf{Actuarial pricing under $\mathbb P$ :} Define $V$ by \eqref{eq:V_discount} with expectation under $\mathbb P$.
The jump operator is
\[
\mathsf J_{\mathbb P}[V](t,\lambda,u)
:=\lambda\!\int_0^\infty\!\big[V(t,\lambda+\beta x,u+x)-V(t,\lambda,u)\big]\nu(\lambda,dx),
\]
so $V$ solves \eqref{eq:PIDE_discount} with $\nu$.
\smallskip
\item \textbf{Esscher pricing on marks :} Fix $\theta$ in the domain of exponential moments such that $k_\theta(\lambda):=\int_0^\infty e^{\theta x}\,\nu(\lambda,dx)<\infty$ for all $\lambda \in \Lambda$.
Following the Esscher approach for jump models \cite{GerberShiu1994,ContTankov2004,BeBe12}, we define the tilted (normalised) mark kernel
\[
\nu_\theta(\lambda,dx)\ :=\ \frac{e^{\theta x}}{k_\theta(\lambda)}\,\nu(\lambda,dx).
\]
By the Girsanov theorem for marked point processes (see, e.g., Cohen \& Elliott~\cite[Cor.~15.3.7]{CohenElliott2015} or Brémaud~\cite{Bremaud1981}), introducing the standard Radon-Nikodym density defined by this tilt yields an equivalent martingale measure $\mathbb Q_\theta$. Assuming standard integrability conditions so that this density is a true martingale on $[0,T]$, the compensator of the jump measure under $\mathbb Q_\theta$ becomes $\lambda_s\,\nu_\theta(\lambda_{s-},dx)\,ds$. 

\noindent Consequently, the value function $V_\theta(t,\lambda,u):=\E_{\mathbb Q_\theta}\!\big[e^{-r(T-t)}f(U_T)\mid \lambda_t=\lambda,U_t=u\big]$ satisfies the same backward PIDE template, but with the jump operator evaluated under the tilted kernel:
\[
\mathsf J_\theta[V](t,\lambda,u)\ =\
\lambda\!\int_0^\infty\!\big[V_\theta(t,\lambda+\beta x,u+x)-V_\theta(t,\lambda,u)\big]\,\nu_\theta(\lambda,dx).
\]
\smallskip
\end{enumerate}

\noindent In both cases the backward structure is identical; only the jump operator changes
(from $\nu$ to $\nu_\theta$). We now state the corresponding PIDE under the Esscher convention.

\medskip
\begin{assumption}[Esscher admissibility]\label{ass:esscher}
There exists $\theta_0>0$ such that for every compact $K\subset\R_+$ and $|\theta|<\theta_0$,
\[
\sup_{\lambda\in K}\int_0^\infty e^{\theta x}\,\nu(\lambda,dx)<\infty,
\qquad
k_\theta(\lambda):=\int_0^\infty e^{\theta x}\,\nu(\lambda,dx)\in(0,\infty)\ \text{for }\lambda\in K.
\]
Moreover $t\mapsto\lambda_{t-}$ is predictable and locally bounded (up to localization). We assume $k_\theta(\lambda)\in(0,\infty)$ and locally bounded and bounded away from 0 on compact $\lambda$-sets.
\end{assumption}

\begin{prop}[Backward PIDE under Esscher pricing]\label{prop:PIDE_Esscher}
Under Assumptions~\ref{ass:pdmp} and \ref{ass:esscher}, $V_\theta(t,\lambda,u):=\E_{\mathbb Q_\theta}[e^{-r(T-t)}f(U_T)\mid \lambda_t=\lambda,U_t=u]$
solves, for $t<T$,
\begin{equation}\label{eq:PIDE_Esscher}
\boxed{\;
\partial_t V_\theta
+\mu(\lambda)\,\partial_\lambda V_\theta
+\lambda\!\int_0^\infty\!\big[V_\theta(t,\lambda+\beta x,u+x)-V_\theta(t,\lambda,u)\big]\nu_\theta(\lambda,dx)
- r\,V_\theta\;=\;0.}
\end{equation}
Equivalently, $(\mathcal A_\theta-r)V_\theta=0$, where $\mathcal A_\theta$ is the extended generator
defined as in \eqref{def:Gen} with $\nu$ replaced by $\nu_\theta$, acting on the common domain
$\mathcal D:=\mathcal D(\mathcal A)\cap\mathcal D(\mathcal A_\theta)$.
\end{prop}

\begin{proof}[Proof sketch]
Girsanov for marked point processes with predictable $Y_s(x)=e^{\theta x}/k_\theta(\lambda_{s-})$
(Brémaud~\cite[Ch.~II]{Bremaud1981}; Jacod--Shiryaev~\cite[III.3]{JacodShiryaev2003}) yields the new compensator
$\lambda_s\nu_\theta(\lambda_s,dx)\,ds$ under $\mathbb Q_\theta$.
Applying Dynkin/Feynman--Kac as in Proposition~\ref{prop:PIDE_discount}
gives \eqref{eq:PIDE_Esscher}.
\end{proof}

\begin{remark}[Stability under OU drift]
An unnormalised tilt $Y(x)=e^{\theta x}$ yields compensator
$\lambda_s e^{\theta x}\nu(\lambda_{s-},dx)\,ds$, i.e. an effective intensity
$\lambda_s k_\theta(\lambda_{s-})$ and a normalised kernel $\nu_\theta$.
We adopt the normalised form to keep the intensity $\lambda$ unchanged.
\end{remark}

\subsection{Interpretation and calibration of $\theta$.}
\noindent The Esscher tilt $\nu_\theta(\lambda,dx)\propto e^{\theta x}\nu(\lambda,dx)$ increases the marks in the likelihood–ratio order,
hence in the increasing convex order; in particular, for convex payoffs one expects prices to increase with~$\theta$
(see, e.g., Shaked–Shanthikumar, \cite{ShakedShanthikumar2007}).
Differentiating the log–partition $k_\theta(\lambda)$ yields
$\partial_\theta\log k_\theta(\lambda)=\E_{\nu_\theta(\lambda)}[X]$ and
$\partial_{\theta\theta}\log k_\theta(\lambda)=Var_{\nu_\theta(\lambda)}[X]\ge0$,
confirming heavier tails as $\theta$ increases. In practice, $\theta$ is calibrated to liquid instruments (e.g., swaps/caps) by matching observed prices under the model with kernel $\nu_\theta$.
This requires a parametric/semiparametric estimate of the actuarial mark law $\nu$
(e.g., a $\lambda$–dependent Gamma mixture via EM), which we detail in the next section.

\subsection{Gamma mixture-of-experts for marks.}

\noindent Marks $X>0$ are positive, skewed, with heavier right tails.
Gamma (and mixtures) are standard for precipitation,
and the Esscher tilt preserves the Gamma family, which makes pricing tractable.
Allowing the mixture weights to depend on the pre-jump intensity $\lambda$ yields
a flexible yet interpretable semi-parametric model $\nu(\lambda,\cdot)$. See Benth \& Šaltytė Benth~\cite[Ch.~VIII]{BeBe12} for precipitations; and Cont-Tankov, §10~\cite{ContTankov2004}, for Esscher in jump models.

\begin{definition}[Gamma mixture-of-experts for the mark law]\label{def:moe}
For $\lambda\in\R_+$, we model
\[
\nu(\lambda,dx)=\sum_{m=1}^{M} \pi_m(\lambda)\,\Gamma(k_m,b_m)(dx),
\qquad
\pi_m(\lambda)=\frac{\exp\{\beta_m^\top B(\lambda)\}}{\sum_{\ell=1}^{M}\exp\{\beta_\ell^\top B(\lambda)\}},
\]
where $B(\lambda)\in\R^p$ is a user-chosen feature map (e.g. $p$ spline basis functions) and $\beta_m\in\R^p$ are gating coefficients with $\beta_M\equiv 0$ (identifiability). Here $p$ (feature dimension) is independent of
$M$ (number of mixture components).
$\Gamma(k_m,b_m)$ is  the Gamma distribution with shape $k_m>0$, rate $b_m>0$, and probability density function of $\Gamma(k_m,b_m)(dx)=\frac{b_m^{k_m}}{\Gamma(k_m)}x^{k_m-1}e^{-b_m x}\,(dx)$ with mean $k_m/b_m$ and Laplace transform
$\hat g_\Gamma(u)=(b_m/(b_m+u))^{k_m}$ for $u\ge0$.
\end{definition}

\paragraph{Estimation by Expectation-Maximization (EM).}
\noindent Given data $\{(\lambda_k,X_k)\}_{k=1}^n$ (pre-jump intensities and marks), parameters
$\Theta=\{(k_m,b_m,\beta_m)\}_{m=1}^M$ are estimated by EM:
(i) E-step: responsibilities $r_{km}\propto \pi_m(\lambda_k)\,g_\Gamma(X_k;k_m,b_m)$ with $g_\Gamma(x;k,b)=b^k x^{k-1}e^{-bx}/\Gamma(k)$;
(ii) M-step: weighted MLE for $(k_m,b_m)$ (closed form for $b_m$ given $k_m$; Newton update for $k_m$ via
$\log k_m-\psi(k_m)=\log\bar X_m-\overline{\log X}_m$, $\psi$ the digamma function), and multinomial logistic regression
for $\beta$ (IRLS/Newton, concave objective). EM increases the log-likelihood and converges to
a stationary point \cite{DempsterLairdRubin1977,McLachlanPeel2000,JordanJacobs1994,HTF2009, Bishop2006}.
\emph{Implementation details} (initialisation, regularisation, stopping rule) follow standard practice and are omitted; see \cite{DempsterLairdRubin1977,McLachlanPeel2000,JordanJacobs1994,Bishop2006,HTF2009}.

\begin{lemma}[Esscher transform of a Gamma MoE]\label{lem:esscher-gamma}
For $\theta$ with $b_m>\theta$ for all $m$, the Esscher-tilted kernel
$\nu_\theta(\lambda,dx)\propto e^{\theta x}\nu(\lambda,dx)$ remains a Gamma mixture:
\[
\nu_\theta(\lambda,dx)=\sum_{m=1}^M \pi_m^\theta(\lambda)\,\Gamma(k_m,b_m-\theta)(dx),
\qquad
\pi_m^\theta(\lambda)=\frac{\pi_m(\lambda)\,(b_m/(b_m-\theta))^{k_m}}{\sum_{\ell=1}^M \pi_\ell(\lambda)\,(b_\ell/(b_\ell-\theta))^{k_\ell}}.
\]
\emph{Proof.} Multiply each component by $e^{\theta x}$ and renormalise using
$\int_0^\infty e^{\theta x}\Gamma(k_m,b_m)(dx)=(b_m/(b_m-\theta))^{k_m}$.
\end{lemma}

\begin{lemma}[Domain of $\theta$]\label{lem:theta-domain}
If the component rates $\{b_m\}_{m=1}^M$ are constant (independent of $\lambda$), then
$\sup_{\lambda\in\Lambda} M_\nu(\lambda,\theta) = \sup_{\lambda\in\Lambda} \int_0^\infty e^{\theta x}\nu(\lambda,dx)<\infty$ if and only if
\[
\theta\ <\ \min_{1\le m\le M} b_m.
\]
\emph{Proof.} $M_\nu(\lambda,\theta)=\sum_m \pi_m(\lambda)\,(b_m/(b_m-\theta))^{k_m}$ is finite iff
$\theta<b_m$ for every active component $m$.
\end{lemma}

\begin{prop}[Swap under Esscher]\label{prop:swap-esscher}
Let $t\le t_1<t_2\le T$ and $r\ge0$. Assuming the payoff is settled at $t_2$, under $\mathbb Q_\theta$ the time-$t$ price of the swap
with payoff $\int_{t_1}^{t_2}\!dU_s$ equals
\[
\mathrm{Swap}_\theta(t)\;=\;e^{-r(t_2-t)}\E_{\mathbb Q_\theta}\!\left[\int_{t_1}^{t_2}\lambda_s\,m_1^{(\theta)}(\lambda_s)\,ds\ \Big|\ \mathcal F_t\right],
\]
where $m_1^{(\theta)}(\lambda)=\int_0^\infty x\,\nu_\theta(\lambda,dx)
=\sum_m \pi_m^\theta(\lambda)\,(k_m/(b_m-\theta))$.
If $(\lambda,U)$ is Markov and $t\le t_1$, then $\mathrm{Swap}_\theta(t)=F(t,\lambda_t)$ for some deterministic $F$.
\end{prop}

\begin{proof}[Proof]
Notice that $\int_{t_1}^{t_2} dU_s=\int_{t_1}^{t_2}\!\!\int_0^\infty x\,N(ds,dx)$ with jump measure $N$.
Under $\mathbb Q_\theta$, the compensator of $N$ is $\lambda_s\,\nu_\theta(\lambda_{s-},dx)\,ds$; for predictable integrand 
$H(s,x)=x$, the compensated integral is a martingale, hence the formula (compensation for marked point processes; see Brémaud~\cite[Ch.~II]{Bremaud1981} or Jacod--Shiryaev~\cite[III.3]{JacodShiryaev2003}).
\[
\E_{\mathbb Q_\theta}\!\left[\int_{t_1}^{t_2}\!\!\int_0^\infty x\,N(ds,dx) \;\Bigg|\;\mathcal F_t\right]
=\E_{\mathbb Q_\theta}\!\left[\int_{t_1}^{t_2}\!\!\int_0^\infty x\,\lambda_s\,\nu_\theta(\lambda_{s-},dx)\,ds \;\Bigg|\;\mathcal F_t\right].
\]
This yields the formula. 
\end{proof}

\subsection{Calibration of $\theta$.}
By Lemma~\ref{lem:esscher-gamma} and Proposition~\ref{prop:swap-esscher}, model prices under
$\nu_\theta$ are smooth (and typically monotone) in $\theta$.
We calibrate $\theta$ to liquid swaps/caps by least-squares fit of model to market prices, computed by Monte Carlo simulation
under the normalised Esscher tilt. Jump times are generated by adaptive thinning for the conditional intensity
(Lewis--Shedler~\cite{LewisShedler1979}; Ogata~\cite{Ogata1981}); under the OU flow
$\dot\lambda=\kappa(\bar\lambda-\lambda)$ one may use the dominating rate
$M=(1+\varepsilon)\max\{\lambda_s,\bar\lambda\}$ between candidate events.
Variance reduction via Common Random Numbers follows Glasserman~\cite[Ch.~IV/V]{Glasserman2004},
and the one-dimensional optimisation is performed by Brent’s method~\cite{Brent1973}.

\begin{remark}[Stability under OU drift]\label{rq:stability-cdtn}
For $\mu(\lambda)=\kappa(\bar\lambda-\lambda)$, a sufficient stability/ergodicity condition is
\[
\kappa\ >\ \beta\,\sup_{\lambda\ge K} m_1^{(\theta)}(\lambda)
\]
for some $K\ge \bar\lambda$, i.e. mean-reversion dominates the expected self-excitation at large $\lambda$.
This matches standard PDMP stability criteria for linear drift (see, e.g., Example~2 in~\cite{ET18} or related results). For the Gamma MoE under normalised Esscher,
$m_1^{(\theta)}(\lambda)=\sum_m \pi_m^\theta(\lambda)\,k_m/(b_m-\theta)\le \max_m k_m/(b_m-\theta)$,
yielding the practical bound $\alpha>\beta\max_m k_m/(b_m-\theta)$.
\end{remark}

\subsection{Spectral reduction in the accumulation variable $u$}

We now reduce the $(t,\lambda,u)$ backward equation to a family of $(t,\lambda)$
equations by exploiting that exponentials are eigenfunctions of the jump–translation $u\mapsto u+x$ (cf. Folland \cite[§7.3]{Folland1992}  and Cont--Tankov~\cite[Ch.~12, §12.2.2]{ContTankov2004}).Throughout this subsection, let $\tilde\nu(\lambda,dx)$ denote the pricing kernel (either actuarial $\nu$ or Esscher $\nu_\theta$). The discounted backward PIDE reads, for $t<T$,
\begin{equation}\label{eq:PIDE_u_lambda_generic}
\partial_t V+\mu(\lambda)\,\partial_\lambda V
+\lambda\!\int_0^\infty\!\big[V(t,\lambda+\beta x,u+x)-V(t,\lambda,u)\big]\tilde\nu(\lambda,dx)
- r\,V=0,\qquad V(T,\lambda,u)=f(u).
\end{equation}

\begin{assumption}[Laplace admissibility]\label{ass:laplace}
There exists $s_0>0$ such that for every compact $K\subset\R_+$, the moment generating function of the mark kernel is finite:
\[
\sup_{\lambda\in K}\int_0^\infty e^{s x}\,\tilde\nu(\lambda,dx) < \infty\qquad\text{for all }s\in[0,s_0).
\]
Furthermore, for the given payoff function $f$, there exists a damping parameter $\delta\in(0,s_0)$ such that
\[
\int_0^\infty e^{-\delta u}\,|f(u)|\,du < \infty.
\]
\end{assumption}

\begin{remark}
Assumption~\ref{ass:laplace} is local in $\lambda$ and is verified by our Gamma–mixture model
(on compacts, the MGF is finite uniformly in $\lambda$). The damping condition on $f$ is the
Laplace–analogue of the Carr--Madan Fourier damping used in option pricing
\cite{CarrMadan1999,ContTankov2004}.
\end{remark}

\begin{prop}[Rank-one exponential ansatz]\label{prop:rankone}
Fix $\eta\in\C$ with $\Re(\eta)<s_0$ and set $V_\eta(t,\lambda,u):=e^{\eta u}\,F(t,\lambda,\eta)$.
Under Assumption~\ref{ass:laplace}, $V_\eta$ solves \eqref{eq:PIDE_u_lambda_generic} with terminal condition $V_\eta(T,\lambda,u)=e^{\eta u}$ if and only if $F(\cdot,\cdot,\eta)$ solves the $(t,\lambda)$ PIDE :
\begin{equation}\label{eq:PIDE_F_eta}
\boxed{\;
\partial_t F+\mu(\lambda)\,\partial_\lambda F
+\lambda\!\int_0^\infty\!\big[e^{\eta x}\,F(t,\lambda+\beta x,\eta)-F(t,\lambda,\eta)\big]\tilde\nu(\lambda,dx)
- r\,F=0,\qquad F(T,\lambda,\eta)=1.}
\end{equation}
\end{prop}

\begin{proof}[Proof]
Plug $V_\eta=e^{\eta u}F$ into \eqref{eq:PIDE_u_lambda_generic}; the jump term factors as
$e^{\eta u}\lambda\int[e^{\eta x}F(t,\lambda+\beta x,\eta)-F(t,\lambda,\eta)]\tilde\nu(\lambda,dx)$.
Divide by $e^{\eta u}\neq 0$.
\end{proof}

\begin{lemma}[Bromwich inversion for the payoff]\label{lem:bromwich}
Let $\delta\in(0,s_0)$ be as in Assumption~\ref{ass:laplace} and define
\[
\widehat f_\delta(y):=\int_0^\infty e^{-(\delta+iy)u}\,f(u)\,du,\qquad y\in\R.
\]

\noindent If $e^{-\delta u}f(u)\in L^1(\R_+)$ and $f$ is piecewise $C^1$ on $\R_+$ with
$e^{-\delta u}f'(u)\in L^1(\R_+)$, then for every $u>0$
\begin{equation}\label{eq:bromwich}
\frac{f(u^-)+f(u^+)}{2}\;=\;\lim_{Y\to\infty}\frac{1}{2\pi}\int_{-Y}^{Y}\widehat f_\delta(y)\,e^{(\delta+iy)u}\,dy,
\end{equation}
In particular, the integral evaluates exactly to $f(u)$ wherever $f$ is continuous. If $f$ has a jump at $u$, \eqref{eq:bromwich} yields the Dirichlet--Jordan midpoint value (see Folland~\cite[Thm.~7.6]{Folland1992}).

\end{lemma}

\begin{remark}
\eqref{eq:bromwich} is the Laplace/Fourier inversion formula written on the vertical contour $s=\delta+iy$, i.e. \ on the line $\Re(s)=\delta$ in the complex plane, parameterized by $y\in\mathbb R$. 
The assumptions $e^{-\delta u}f(u)\in L^1(\mathbb R_+)$, piecewise $C^1$ regularity of $f$, and
$e^{-\delta u}f'(u)\in L^1(\mathbb R_+)$ ensure the usual Dirichlet--Jordan midpoint convergence at jump points.
Truncated spectral inversion of non-smooth payoffs (e.g., digitals) exhibits Gibbs oscillations; Fejér/Cesàro averaging converges to the Dirichtlet-Jordan value at jumps, and spectral filters (e.g. exponential filters) reduce ringing.
See \cite{Zygmund2002, Trefethen2000, Boyd2001, GottliebShu1997, Vandeven1991}.
\end{remark}

\begin{theorem}[Superposition/inversion across the Bromwich line]\label{thm:inversion}
Suppose Assumption~\ref{ass:laplace} holds and fix $\delta\in(0,s_0)$. For each $y\in\R$, let $F(\cdot,\cdot,\delta+iy)$
be a classical solution of \eqref{eq:PIDE_F_eta} with terminal condition equal to $1$, and assume the polynomial growth bound, for every compact $K\subset\R_+$ and some $p\ge0$, 
\[
\sup_{(t,\lambda)\in [0,T]\times K}\,\big(|F| + |\partial_t F| + |\partial_\lambda F|\big)(t,\lambda,\delta+iy)\ \le\ C_K\,(1+|y|)^p
\]
and the weighted integrability, assuming for $\delta$, $\forall K \subset \R $ $\sup_{\lambda\in K}\sup_{u\ge0}\E_{\lambda,u}\!\left[e^{\delta(U_T-u)}\right]<\infty$ 
\[
\int_{\R}\!|\widehat f_\delta(y)|(1+|y|)^p\,dy<\infty.
\]
Then the function
\begin{equation}\label{eq:V_reconstruction}
\boxed{\;
V(t,\lambda,u)\ =\ \frac{1}{2\pi}\int_{\R}\widehat f_\delta(y)\,e^{(\delta+iy)u}\,F\big(t,\lambda,\delta+iy\big)\,dy}
\end{equation}
is well defined, belongs to $\mathcal D(\mathcal A)$ with growth $|V|\le C\,e^{\delta u}$, solves
\eqref{eq:PIDE_u_lambda_generic} in the classical sense, and satisfies $V(T,\lambda,u)=f(u)$ for $u\ge0$.
Moreover, $V$ is the unique solution in the class $\{\,|V|\le C\,e^{\delta u}\,\}$.
\end{theorem}

\begin{proof}
Let $V_y(t,\lambda,u) := \widehat f_\delta(y)\,e^{(\delta+iy)u}\,F(t,\lambda,\delta+iy)$ denote the integrand.
Recall the linear operator $L$ associated with \eqref{eq:PIDE_u_lambda_generic}: $L\phi := \partial_t\phi+\mu(\lambda)\partial_\lambda\phi +\lambda\int \Delta_x \phi\, \tilde\nu(\lambda,dx) - r\phi$,
where $\Delta_x \phi = \phi(t,\lambda+\beta x,u+x)-\phi(t,\lambda,u)$.
By Proposition~\ref{prop:rankone}, for each fixed $y$,  $e^{(\delta+iy)u}F(\cdot,\cdot,\delta+iy)$ is a solution, hence $L(V_y) = 0$.

\emph{(i) Well-definedness and Regularity :} By the polynomial bound in $y$ and the weighted integrability assumption, the integrand $V_y$ and its partial derivatives $\partial_t V_y, \partial_\lambda V_y$ are dominated by an integrable function of $y$ (locally in $(t,\lambda)$). By the Dominated Convergence Theorem, $V$ is well-defined, continuous, and differentiable with respect to $t$ and $\lambda$, with derivatives obtained by differentiating under the integral sign: $\partial_t V = \frac{1}{2\pi}\int_\R \partial_t V_y\,dy, \quad
\partial_\lambda V = \frac{1}{2\pi}\int_\R \partial_\lambda V_y\,dy.$

\emph{(ii) Equation :} 
The nonlocal term is handled similarly using Fubini's theorem (justified by the absolute convergence of the integral against $\tilde\nu$). By linearity of the integral, we compute $L V$:
\begin{align*}
L V(t,\lambda,u)
&= L \left( \frac{1}{2\pi}\int_\R V_y(t,\lambda,u)\,dy \right) \\
&= \frac{1}{2\pi}\int_\R L \big( V_y(t,\lambda,u) \big)\,dy \quad \text{(by dominated convergence and linearity)} \\
&= \frac{1}{2\pi}\int_\R 0 \,dy \;=\; 0.
\end{align*}

\emph{(iii) Terminal Condition and Uniqueness :} At $t=T$, $F\equiv1$ and Lemma \ref{lem:bromwich} yields
$V(T,\lambda,u)=(2\pi)^{-1}\int \widehat f_\delta(y)\,e^{(\delta+iy)u}\,dy=f(u)$. 

\noindent For uniqueness, let $\widetilde V$ be another classical solution of
\eqref{eq:PIDE_u_lambda_generic} with the same terminal data and such that
$|\widetilde V(t,\lambda,u)|\le C e^{\delta u}$ on $[0,T]\times K\times\R_+$ for every compact $K\subset\R_+$.
Set $H:=V-\widetilde V$. Then $H(T,\cdot)=0$ and $H$ solves $L H=0$.
Let $(X_s)_{s\in[t,T]}=(\lambda_s,U_s)$ be the state process and define a localisation
$\tau_n:=\inf\{s\ge t:\lambda_s\notin K_n \text{ or } U_s\ge n\}\wedge T$, with $K_n\uparrow\R_+$ compact.
By Dynkin's formula for PDMPs applied to the stopped process
$s\mapsto e^{-r(s-t)}H(s,X_s)$ (Davis~\cite[Ch.~3, \S 32]{Davis93}), $H(t,x)=\E_x\!\left[e^{-r(\tau_n-t)}H(\tau_n,X_{\tau_n})\right], x=(\lambda,u)$.
Since $U$ is nondecreasing, $|H(\tau_n,X_{\tau_n})|\le C e^{\delta U_{\tau_n}}\le C e^{\delta U_T}$.
Moreover, under Assumption~\ref{ass:laplace} with $\delta<s_0$ and on each compact $K_n$,
$\sup_{x\in K_n\times\R_+}\E_x[e^{\delta(U_T-u)}]<\infty$ (Laplace/Campbell functional; see
Br\'emaud~\cite[Ch.~II]{Bremaud1981} or Cont--Tankov~\cite[\S10.2]{ContTankov2004}),
so $e^{-r(\tau_n-t)}H(\tau_n,X_{\tau_n})$ is uniformly integrable.
Letting $n\to\infty$ and using $\tau_n\uparrow T$ a.s., dominated convergence yields $H(t,x)=\E_x\!\left[e^{-r(T-t)}H(T,X_T)\right]=0$.
Hence $V=\widetilde V$, i.e. uniqueness holds in the class $\{|V|\le C e^{\delta u}\}$.


\end{proof}

\begin{remark}[Why exponentials?]
For the translation operator $(T_x\phi)(u):=\phi(u+x)$ one has
$T_x\big(e^{\eta u}\big)=e^{\eta x}e^{\eta u}$. Hence exponentials diagonalise translations:
the jump term in $u$ becomes a multiplication by $e^{\eta x}$ after a Laplace/Fourier transform in $u$.
In exponential‐Lévy models this is precisely $L=\psi(\partial)$ with $\psi$ the
characteristic exponent (Cont--Tankov~\cite[Ch.~12, §12.2.2]{ContTankov2004}).
The damping/inversion on a vertical line is the Laplace analogue of Carr--Madan’s Fourier method
\cite{CarrMadan1999,ContTankov2004}.
\end{remark}

\paragraph{Numerical plan (Bromwich--IMEX--Gauss--Laguerre).}
Fix $\delta\in(0,s_0)$ and consider the Bromwich line $\eta=\delta+i y$.
For each $y\in[0,Y_{\max}]$ we solve the reduced $(t,\lambda)$ PIDE \eqref{eq:PIDE_F_eta}
backward in time on a uniform grid $\lambda_i=i\Delta\lambda$:
the OU transport term and discounting are treated implicitly (IMEX in time),
while the nonlocal jump operator is treated explicitly.
The $x$--integral is approximated by generalized Gauss--Laguerre quadrature for each Gamma component of the mark mixture.
Off--grid evaluations at $\lambda_i+\beta x$ are computed by piecewise--linear interpolation and clamped to $[0,\lambda_{\max}]$.

\smallskip
\noindent Once $F(0,\lambda_0,\delta+i y)$ is obtained on a grid $y_j=j\Delta y$, $j=0,\dots,N_y-1$,
the option price is recovered using the inversion formula of Theorem~\ref{thm:inversion} in the one-sided form 
\[
V(0,\lambda_0,u_0)\;\approx\;\frac{e^{\delta u_0}}{\pi}\int_{0}^{Y_{\max}}
\Re\!\Big(\widehat f_\delta(y)\,F(0,\lambda_0,\delta+i y)\,e^{i y u_0}\Big)\,dy,
\]
by conjugate symmetry (since the payoff and PIDE coefficients are real), where the integral is truncated to $[0,Y_{\max}]$ and evaluated by a composite Simpson rule on the $y$--grid.

\paragraph{Choice of $\delta$, $Y_{\max}$ and grids.}
Under the Gamma mixture model (actuarial or Esscher-tilted), exponential moments are finite provided
$0<\delta<\min_m b_m^*$ (with $b_m^*=b_m$ actuarial and $b_m^*=b_m-\theta$ Esscher),
so we choose such a $\delta$ to ensure the finiteness of the jump term in \eqref{eq:PIDE_F_eta}.
The truncation level $Y_{\max}$ and the number of Simpson points $N_y$ are selected to strictly control both the oscillatory tail and the quadrature errors, as empirically validated in our convergence diagnostics (see Section~\ref{subsec:bromwich_diagnostics}).
Finally, the intensity domain $[0,\lambda_{\max}]$ is chosen sufficiently large so that off-grid queries for $\lambda+\beta x$ remain negligible; this is continuously monitored and validated \emph{a posteriori} via the boundary-hit diagnostic introduced in our experimental setup.

\section{IMEX discretization for the $(t,\lambda)$ PIDE at fixed frequency}\label{sec:IMEX}

Fix $\eta=\delta+iy$ with $\delta>0$ from Assumption~\ref{ass:laplace}.
We solve the PIDE \eqref{eq:PIDE_F_eta} for $F=F(t,\lambda,\eta)$ backwards from $T$ to $0$. Exploiting the fact that $\tilde\nu(\lambda,\cdot)$ is a probability kernel (\(\int_0^\infty \tilde\nu(\lambda,dx)=1\)), we adopt an Implicit-Explicit (IMEX) split with the local drift/reaction treated implicitly and the nonlocal jump operator treated explicitly (see Cont-Voltchkova~\cite{ContVoltchkova2005}):

\begin{equation}\label{eq:PIDE_IMEX_form}
\partial_t F\;+\;L_{\mathrm{imp}}F\;+\;N_{\mathrm{exp}}[F]\;=\;0,\qquad F(T,\lambda,\eta)=1.
\end{equation}
with 
\begin{equation}\label{eq:IMEX_split}
L_{\mathrm{imp}}F \ :=\ \mu(\lambda)\,\partial_\lambda F\ -\ r\,F,
\qquad
N_{\mathrm{exp}}[F]\ :=\ \lambda\!\int_0^\infty\!\Big(e^{\eta x}F(t,\lambda+\beta x,\eta)-F(t,\lambda,\eta)\Big)\,\tilde\nu(\lambda,dx).
\end{equation}
By doing so, if $\eta=0$ and $F \equiv 1$ spatially, $N_{\mathrm{exp}}[F]$ vanishes exactly, avoiding any spurious spatial numerical diffusion for the zero-coupon bond mode.

\subsection{Grids and discrete operators}
Let $t_n := n\,\Delta t$ $(n=0,\dots,N)$ with $t_N=T$ and $\lambda_i:=\lambda_{\min}+i\,h$ ($i=0,\dots,I$). We solve backward from $t_N=T$ to $t_0=0$.
We denote $F_i^n:=F(t_n,\lambda_i,\eta)$ and set $\mu_i=\mu(\lambda_i)$, 
$\mu_i^+=\max(\mu_i,0)$, $\mu_i^-=\min(\mu_i,0)$.   

\noindent To handle the convection term $\mu(\lambda)\partial_\lambda$, we employ a first-order upwind difference scheme and because the flow of time is reversed in the backward PDE ($\tau = T-t$), the effective advection velocity is $-\mu(\lambda)$. Although only first-order accurate, the upwind scheme is monotone and prevents spurious oscillations (such as negative prices) that typically arise with central differences in convection-dominated regimes or with non-smooth initial data. It also gives us an M-matrix, useful for the discrete maximum principle (see LeVeque~\cite[§3.1–3.3]{LeVeque2002}) :
\[
(D^-F)^n_i:=\frac{F_i^n-F_{i-1}^n}{\Delta\lambda},\qquad
(D^+F)^n_i:=\frac{F_{i+1}^n-F_i^n}{\Delta\lambda}.
\]
and the implicit local operator is :
\begin{equation}\label{eq:Limp_discrete}
(L_{\mathrm{imp}}^\Delta F^n)_i
:=\ \mu_i^+\,(D^+F)^n_i\;+\;\mu_i^-\,(D^-F)^n_i\;-\;r\,F_i^n.
\end{equation}

\noindent The explicit nonlocal operator is $(N_{\mathrm{exp}}^\Delta[F^{n+1}])_i \;=\; \lambda_i\big(Q_i^{n+1}-F_i^{n+1}\big)$,
where $Q_i^{n+1}$ is defined by \eqref{eq:Q_def}.
\begin{equation}\label{eq:Q_def}
Q_i^{n+1}
:= \int_0^\infty e^{\eta x}\,
\Pi_{\Delta\lambda}\!\left[F^{n+1}\right](\lambda_i+\beta x)\,\tilde\nu(\lambda_i,dx),
\qquad \eta=\delta+iy.
\end{equation}
and $\Pi_{\Delta\lambda}$ denotes a monotone interpolation on the $\{\lambda_i\}$ grid (PCHIP to preserve positivity and avoid spurious oscillations) and we evaluated the integral by a generalized Gauss-Laguerre scheme.

\paragraph{Boundary conditions in $\lambda$.}
At $\lambda_{\min}=0$ we impose a homogeneous Neumann condition $\partial_\lambda F=0$,
implemented by a ghost point $(D^-F)^n_0=0$.
At $\lambda_{\max}$ we extrapolate monotonically.
The truncation level $\lambda_{\max}$ is chosen so that evaluations of the shifted argument
$\lambda_i+\beta x$ in \eqref{eq:Q_def} rarely fall outside $[0,\lambda_{\max}]$:
we monitor the (quadrature-weighted) boundary-hit ratio and increase $\lambda_{\max}$ until it is below a prescribed tolerance.

\subsection{The backward IMEX step}
A first-order IMEX–Euler step from $t_{n+1}$ to $t_n$ reads
\begin{equation}\label{eq:IMEX_step}
\frac{F_i^{n+1}-F_i^{n}}{\Delta t} + (L_{\mathrm{imp}}^\Delta F^n)_i + (N_{\mathrm{exp}}^\Delta[F^{n+1}])_i = 0.
\end{equation}
Equivalently, for $n=N-1,\dots,0$ we have the IMEX scheme :
\begin{equation}\label{eq:linear_system}
\boxed{\big(I-\Delta t\,L_{\mathrm{imp}}^\Delta\big)F^n \;=\; F^{n+1}\;+\;\Delta t\,N_{\mathrm{exp}}^\Delta[F^{n+1}]}
\end{equation}
with terminal data $F^N\equiv\mathbf{1}$ at $t_N=T$.

\paragraph{Tri-diagonal structure.}
\noindent Writing \eqref{eq:Limp_discrete} out, \eqref{eq:linear_system} is a tri-diagonal system \(\mathbf A F^n = \mathbf b^{n+1}\) with coefficients below and $F^n:=(F_0^n,\dots,F_I^n)^\top$ the vector of nodal values at time $t_n$:

\[
A_{i,i-1}\,F_{i-1}^n + A_{i,i}\,F_i^n + A_{i,i+1}\,F_{i+1}^n = \mathrm{RHS}_i^{n+1}\qquad (1\le i\le I-1),
\]
with coefficients
\begin{equation}\label{eq:coeff_A}
A_{i,i}=1+\Delta t\!\left(\frac{\mu_i^+-\mu_i^-}{\Delta\lambda}+r\right),\quad
A_{i,i-1}=\Delta t\,\frac{\mu_i^-}{\Delta\lambda}\le 0,\quad
A_{i,i+1}=-\Delta t\,\frac{\mu_i^+}{\Delta\lambda}\le 0,
\end{equation}
and right-hand side
\begin{equation}\label{eq:coeff_RHS}
\mathrm{RHS}_i^{n+1}
= F_i^{n+1} + \Delta t\,\lambda_i\big(Q_i^{n+1}-F_i^{n+1}\big)
= \bigl(1 - \Delta t\,\lambda_i\bigr)F_i^{n+1} + \Delta t\,\lambda_i Q_i^{n+1}.
\end{equation}

We solve this system by the Thomas (LU tridiagonal) algorithm in $\mathcal O(I)$ per time step and per frequency, see, e.g., Press and al.\cite[§2.4]{PressTeukolskyVetterlingFlannery2007}. For M-matrix diagonal dominant from the implicit upwind, under the CFL condition we're going to see below the process is numerically stable.


\subsection{Gauss--Laguerre approximation of the jump integral}

With the Gamma mixture (see Definition~\ref{def:moe}),
$\tilde\nu(\lambda_i,dx)=\sum_{m=1}^M \pi_m(\lambda_i)\,\Gamma(k_m,b_m)(dx)$, after the change of variable $z=b_m x$ in \eqref{eq:Q_def} we obtain the representation
\begin{equation}\label{eq:Q_after_change}
Q_i^{n+1}
= \sum_{m=1}^M \pi_m(\lambda_i)\,\frac{1}{\Gamma(k_m)}
\int_0^\infty e^{\eta z/b_m} \Pi_{\Delta\lambda}\!\left[F^{n+1}\right]\Bigl(\lambda_i+\beta \tfrac{z}{b_m},\eta\Bigr)
z^{k_m-1} e^{-z}\,dz.
\end{equation}

\noindent Since $b_m>\Re(\eta)$ for all $m$ under the chosen damping (cf. Lemma~\ref{lem:theta-domain}), we have
$|e^{\eta z/b_m} e^{-z}| = e^{-(1-\Re(\eta)/b_m)z}$, so the integrand decays exponentially.
Moreover $k_m-1>-1$, hence the weight $z^{k_m-1}e^{-z}$ is integrable on $(0,\infty)$.
Therefore the integral in \eqref{eq:Q_after_change} can be approximated by the
Generalized Q-point Gauss--Laguerre rule with parameter $\alpha_m=k_m-1$ :

\begin{equation}\label{eq:Q_discrete_GL}
Q_i^{n+1}
\;\approx\;
\sum_{m=1}^M \pi_m(\lambda_i)\,\frac{1}{\Gamma(k_m)}
\sum_{q=1}^Q w_{m,q}^{(k_m-1)}\,
e^{\eta z_{m,q}/b_m}\,
\Pi_{\Delta\lambda}\!\left[F^{n+1}\right]\Bigl(\lambda_i+\beta \tfrac{z_{m,q}}{b_m},\eta\Bigr),
\end{equation}
where $\{z_{m,q},w_{m,q}^{(k_m-1)}\}_{q=1}^Q$ are the Gauss--Laguerre (GL) nodes/weights for parameter $\alpha=k_m-1$, computed by Golub–Welsch~\cite{GolubWelsch1969} (i.e. by diagonalising the
$Q\times Q$ Jacobi matrix associated with the Laguerre weight
$z^{k_m-1}e^{-z}$); cf. 
Davis–Rabinowitz~\cite[Ch.~3.6]{DavisRabinowitz1984}. In practice, $Q=16$ or $32$ are enough; we control $\|Q^{(Q)}-Q^{(2Q)}\|$ and refine if necessary (see \cite{DavisRabinowitz1984} for error bounds). If $\lambda_i+\beta z_{m,q}/b_m\notin[\lambda_{\min},\lambda_{\max}]$ we apply the monotone extrapolation and monitor the boundary-hit ratio.

This choice is motivated by three factors : (i) It handles the semi-infinite domain $[0,\infty)$ without artificial truncation and with the exponential decrease it is optimal; (ii) The weights are positive so no spurious oscillations when combined with a monotone interpolant (PCHIP); (iii) The nodes/weights are  independent of $t$ and of the $\lambda$-grid so they can be pre-computed per components. The structure \eqref{eq:Q_def} still holds for a general kernel $\tilde\nu$ (Cf. Rem \ref{rem:generic-kernel})

\begin{remark}[Kernel-agnostic variant]\label{rem:generic-kernel}
If $\tilde\nu(\lambda,dx)$ isn't a Gamma-mixture, we keep \eqref{eq:Q_def} and we replace the GL with
a positive quadrature adapted to the support (e.g. Gauss–Jacobi after reparametrization, or adaptive Gauss–Kronrod).
The IMEX scheme and tri-diagonal structure do not change; only the integration routine varies.
\end{remark}

\subsection{Shape-preserving interpolation in $\lambda$ (PCHIP)}

In the Gauss--Laguerre approximation \eqref{eq:Q_discrete_GL}, we must evaluate $F^{n+1}(\lambda_i+\beta z_{m,q}/b_m,\eta)$ at off-grid points $\lambda' := \lambda_i + \beta \frac{z_{m,q}}{b_m}$, while it is only known on the spatial grid $\{\lambda_j\}_{j=0}^I$. We use the shape-preserving Piecewise Cubic Hermite interpolant (PCHIP) $\Pi_{\Delta\lambda}[F^{n+1}]$ built on $\{(\lambda_j, F_j^{n+1})\}_{j=0}^I$ with the Fritsch-Carlson slope \cite{FritschCarlson1980,FritschButland1984} (with monotone endpoint slope limiting \cite{FritschButland1984}).

\medskip 
We employ PCHIP because unlike standard cubic splines whose interpolant is $C^2$, which are prone to overshoot on monotone data (and may violate nonnegativity), PCHIP constructs a $C^1$ piecewise-cubic interpolant and enforces
local monotonicity via the Fritsch–Carlson slope limiting
\cite{FritschCarlson1980,FritschButland1984}. Concretely, nodal derivatives are set as
weighted harmonic means of adjacent secant slopes, which guarantees that on any cell
where the nodal data are monotone the interpolant is also monotone and stays between
the two endpoint values. In particular, if the nodal values are nonnegative, the
interpolant remains nonnegative on each monotone cell. This avoids the overshoots
observed with classical cubic splines \cite[§VI]{deBoor1978} and is substantially
less diffusive than piecewise linear interpolation.

\noindent From a numerical standpoint, this choice pairs well with the positive Gauss–Laguerre
weights used in the jump integral: “no-overshoot” interpolation plus positive quadrature
helps prevent spurious creation of mass in the nonlocal term and improves the robustness
of the IMEX time stepping. For complex-valued modes $F(\cdot,\delta+iy)$, we apply PCHIP
componentwise to the real and imaginary parts, which preserves the shape constraints
on each component \cite{Hyman1983}. While PCHIP provides $\mathcal{O}(\Delta\lambda^2)$ spatial accuracy, simpler piecewise-linear interpolation may also be used in practice to ensure strict 1-Lipschitz stability at the cost of first-order accuracy.

\bigskip

\subsection{Spatial interpolation error}

We quantify the error introduced by the off-grid evaluation of $F^{n+1}$. While standard for monotone Hermite interpolants (see \cite{FritschCarlson1980,Hyman1983,FritschButland1984}; cf. also
\cite[Ch.~II]{deBoor1978}), this estimate is required for the global convergence analysis.

\begin{assumption}[Smoothness and mesh]\label{ass:PCHIP_smooth}
Let $F\in C^2([\lambda_{\min},\lambda_{\max}])$ and let
$\lambda_0<\dots<\lambda_I$ be a partition with steps
$h_j:=\lambda_{j+1}-\lambda_j$ and $h:=\max_j h_j$.
Denote $F_j:=F(\lambda_j)$ and let $\Pi_{\Delta\lambda}$ be the PCHIP interpolant built from $\{(\lambda_j,F_j)\}_{j=0}^I$ by the Fritsch--Carlson slope
selection.
\end{assumption}

\begin{lemma}[Interpolation error]\label{lem:PCHIP_error}
Under Assumption~\ref{ass:PCHIP_smooth}, there exists a constant $C>0$ depending only on $\|F''\|_{\infty,\Lambda}$ and mesh regularity and $\widetilde F :=\Pi_{\Delta\lambda}[F]$ satisfies
\begin{equation}\label{eq:PCHIP_error}
\|\widetilde F - F\|_{\infty,\Lambda}
\;\le\; C\,h^p.
\end{equation}
\noindent where $p=2$ for PCHIP and $p=1$ for piecewise-linear interpolation.
\end{lemma}
\begin{proof}[Proof Sketch]
For the PCHIP scheme, the error of the cubic Hermite representation is driven by the accuracy of the local slope approximations $m_j$. In Appendix~\ref{app:PCHIP_error_detail}, we prove that the Fritsch--Carlson slope limiters guarantee $|m_j - F'(\lambda_j)| = \mathcal{O}(h)$ in both monotone and non-monotone regions, which yields the global $\mathcal{O}(h^2)$ bound.
For piecewise-linear interpolation, a Taylor expansion on each cell yields
$\| \Pi_{\Delta\lambda}^{\rm lin}F - F\|_{\infty,\Lambda}\le C h\,\|F'\|_{\infty,\Lambda}$, hence \eqref{eq:PCHIP_error} holds with $p=1$.

\end{proof}

\begin{lemma}[Sup–norm Lipschitz property of the nonlocal operator]\label{lem:Q-Lipschitz}
Let $N_{\exp}$ the non-local operator defined in (\ref{eq:IMEX_split}). Under Assumptions \ref{ass:laplace} ($\Re(\eta)=\delta < s_0$), for any bounded $F,G:[0,\infty)\to\C$ on $\Lambda$, always with the clamping convention :
\[
\|N_{\exp}[F]-N_{\exp}[G]\|_{\infty,\Lambda}
\ \le\ \lambda_{\max}\,(C_\delta^*(\Lambda)+1)\ \|F-G\|_{\infty,\Lambda}.
\]
\end{lemma}

\begin{proof}
For $\lambda\in\Lambda$ and since $\tilde\nu(\lambda, \cdot)$ is a probability measure (integrating to 1),
\[
|(N_{\exp}[F]-N_{\exp}[G])(\lambda)|
\le \lambda\!\int_0^\infty\! e^{\delta x}\,|F(\lambda+\beta x)-G(\lambda+\beta x)|\,\tilde\nu(\lambda,dx) \;+\; \lambda\,|F(\lambda)-G(\lambda)|\!\int_0^\infty\!\tilde\nu(\lambda,dx).
\]
Taking the supremum over \(\lambda\in\Lambda\) and using the clamping convention in Section \ref{sec:notations} with $\|f-g\|_{\infty,\Lambda}=\|E_{\rm clamp}f-E_{\rm clamp}g\|_{\infty,[0,\infty)}$ yields the claim.
\end{proof}
\begin{corollary}[Error propagation in the jump operator]\label{cor:PCHIP_jump_error}
Let $\Pi_{\Delta\lambda}F$ be the spatial interpolant (e.g., PCHIP or linear) of $F$ and let $N_{\exp}[F]$ be defined as in \eqref{eq:IMEX_split}, with $\Re(\eta)=\delta<s_0$ and $\tilde\nu(\lambda,dx)$ satisfying the Laplace admissibility in Assumption~\ref{ass:laplace}. Then on the computational domain $\Lambda$,
\[
\bigl\|\,N_{\exp}[\Pi_{\Delta\lambda}F]-N_{\exp}[F]\,\bigr\|_{\infty,\Lambda}
\;\le\; \lambda_{\max}\big(C_\delta^*(\Lambda) + 1\big)\;\|\Pi_{\Delta\lambda}F-F\|_{\infty,\Lambda}
\;=\; \mathcal O(h^p),
\]
uniformly on $\Lambda$, where $C_\delta^*(\Lambda)<\infty$ by Assumption \ref{ass:laplace} and $p$ is the interpolation order.
\end{corollary}

\begin{proof}
    Follows directly from Lemma~\ref{lem:PCHIP_error} and Lemma~\ref{lem:Q-Lipschitz}.
\end{proof}

\begin{prop}[Boundedness of the exponential modes]\label{prop:mode-bounded}
Fix $\eta=\delta+iy$ with $\delta\in[0,s_0)$ as in Lemma~\ref{lem:Q-Lipschitz}. Consider the modal
value function
\[
F(t,\lambda,\eta)
:= \E\Bigl[ e^{-r(T-t)} \prod_{t< T_k \le T} e^{\eta X_k} \,\Big|\, \lambda_t=\lambda \Bigr],
\]
where $\{(T_k,X_k)\}$ is the marked point process of jumps of $(\lambda,U)$ with compensator
$\lambda_s \tilde\nu(\lambda_s,dx)\,ds$ under the chosen pricing measure (actuarial or Esscher). Working on the truncated domain $\Lambda$ (or after localization) so that $\lambda_s\le\lambda_{\max}$ for $s\in[t,T]$, we have the uniform bound

\begin{equation}\label{eq:mode-bound}
|F(t,\lambda,\eta)|
\;\le\;
\exp\!\Bigl( (T-t)\bigl[\lambda_{\max}(M_\eta^\ast - 1) - r \bigr] \Bigr),
\qquad
M_\eta^\ast := \sup_{\lambda\in\Lambda} \int_0^\infty e^{\delta x}\,\tilde\nu(\lambda,dx),
\end{equation}
\end{prop}

\begin{proof}
Taking absolute values and using $|e^{\eta X_k}|=e^{\delta X_k}$, we bound
\[
|F(t,\lambda,\eta)| 
\;\le\; e^{-r(T-t)}\, \E\Bigl[ \prod_{t< T_k \le T} e^{\delta X_k} \,\Big|\, \lambda_t=\lambda \Bigr].
\]
Let $Z_s := \prod_{t< T_k \le s} e^{\delta X_k}$ for $s \ge t$. Denote $A_s := \int_t^s \lambda_v \int_0^\infty \big(e^{\delta x}-1\big)\,\tilde\nu(\lambda_v,dx)\,dv, s\in[t,T]$.
By the marked point process exponential formula (see, e.g., Brémaud~\cite[Ch.~II]{Bremaud1981}), the Doléans--Dade exponential $M_s := Z_s \exp(-A_s)$
is a \emph{positive local martingale}, hence a supermartingale. Therefore
\[
\E[M_T \mid \mathcal F_t] \le M_t = 1,
\qquad\text{and in particular}\qquad
\E[M_T \mid \lambda_t=\lambda]\le 1.
\]

\noindent On the truncated domain (or after localization), $\lambda_v\le \lambda_{\max}$ for $v\in[t,T]$.
Moreover, since $\delta<s_0$, $M_\eta^\ast:=\sup_{\lambda\in\Lambda}\int_0^\infty e^{\delta x}\tilde\nu(\lambda,dx)<\infty$.
Thus for all $v\in[t,T]$,
\[
\lambda_v \int_0^\infty (e^{\delta x}-1)\tilde\nu(\lambda_v,dx)
\;\le\; \lambda_{\max}\Big(\int_0^\infty e^{\delta x}\tilde\nu(\lambda_v,dx)-1\Big)
\;\le\; \lambda_{\max}(M_\eta^\ast-1),
\]
so that $A_T\le (T-t)\lambda_{\max}(M_\eta^\ast-1)$. Hence $Z_T = M_T e^{A_T}
\;\le\; M_T \exp\!\big((T-t)\lambda_{\max}(M_\eta^\ast-1)\big)$.
Taking conditional expectations given $\lambda_t=\lambda$ and using $\E[M_T\mid \lambda_t=\lambda]\le 1$ gives
\[
\E[Z_T\mid \lambda_t=\lambda]
\;\le\; \exp\!\big((T-t)\lambda_{\max}(M_\eta^\ast-1)\big).
\]
Multiplying by the discount factor $e^{-r(T-t)}$ yields \eqref{eq:mode-bound}.

\end{proof}

\begin{remark}[Why this matters]
(i) The bound in Proposition~\ref{prop:mode-bounded} gives \emph{uniform} (in $\lambda\in\Lambda$) control of each modal solution at fixed $\eta$, which we use for dominated convergence and to justify exchanging the $y$–integral with $\partial_t$, $\partial_\lambda$ and the jump integral in Theorem~\ref{thm:inversion}.
(ii) Combined with Lemma~\ref{lem:Q-Lipschitz} and the PCHIP error $\|\Pi_{\Delta\lambda}[F]-F\|_{\infty,\Lambda}=O(h^2)$, it yields a clean $O(h^2)$ perturbation bound for the nonlocal term, feeding into the global convergence estimate.
(iii) Numerically, bounded modal amplitudes help keep the IMEX right-hand side stable across frequencies.
\end{remark}

\subsection{CFL and stability of the IMEX scheme}

Recall the backward IMEX step at a fixed frequency $\eta=\delta+iy$ in the equation (\ref{eq:linear_system}) with $L_{\mathrm{imp}}^\Delta$ given by \eqref{eq:Limp_discrete} (upwind implicit
for $\mu(\lambda)\partial_\lambda-r$) and $N_{\mathrm{exp}}^\Delta$
by \eqref{eq:Q_def} (explicit jump gain via Gauss--Laguerre\,+\,PCHIP).

\noindent To control the backward sweep $t_{n+1}\to t_n$ we need (i) an $\ell^\infty$–bound for the inverse of the implicit matrix $A := I - \Delta t\,L_{\mathrm{imp}}^\Delta$,
and (ii) a CFL–type condition on the explicit jump part to ensure positivity and bounded Lipschitz propagation.
We start with the implicit part.

\begin{prop}[M--matrix property of the implicit part]\label{prop:M-matrix}
Let $A=I - \Delta t\,L_{\mathrm{imp}}^\Delta$ using the coefficients defined in \eqref{eq:coeff_A}. For any time step $\Delta t > 0$, $A$ is a tri-diagonal Z-Matrix that is strictly diagonally dominant by rows. Hence, it is a nonsingular M--matrix (see, e.g., \cite{Varga2000,BermanPlemmons1994}), and its inverse satisfies the Varah bound \cite{Varah1975}:
\begin{equation}\label{eq:Varah}
\boxed{\;
\|A^{-1}\|_{\infty}
\;\le\; \frac{1}{\displaystyle \min_{0\le i\le I}\bigl(1+\Delta t\,r\bigr)}
\;=\; \frac{1}{1+\Delta t\,r} \;\le\; 1.
\;}
\end{equation}
\end{prop}

\begin{proof}
By construction, $A$ is tri-diagonal with nonpositive off-diagonal entries ($A_{i,i-1} \le 0$ and $A_{i,i+1} \le 0$), making it a Z-matrix. For the diagonal dominance, we sum their absolute values:
\[
|A_{i,i-1}| + |A_{i,i+1}|
= \Delta t\,\frac{\mu_i^-}{\Delta\lambda} + \Delta t\,\frac{\mu_i^+}{\Delta\lambda}
= \Delta t\,\frac{\mu_i^+ - \mu_i^-}{\Delta\lambda}.
\]
Subtracting this from the diagonal entry $A_{i,i} = 1 + \Delta t\Bigl(\frac{\mu_i^+ - \mu_i^-}{\Delta\lambda} + r\Bigr)$ yields the margin of diagonal dominance:
\begin{align*}
A_{i,i} - \bigl(|A_{i,i-1}| + |A_{i,i+1}|\bigr)
&= 1 + \Delta t\,r.
\end{align*}
Since $r \ge 0$, this margin is strictly positive (and $\ge 1$) for any $\Delta t > 0$. A strictly diagonally dominant Z-matrix is a nonsingular M--matrix, and the $\ell^\infty$–bound \eqref{eq:Varah} follows directly from the inverse of the minimal row margin.
\end{proof}

\begin{prop}[Explicit CFL for the jump gain]\label{prop:CFL-exp}
Let $\eta=\delta+iy$ with $\delta\in[0,s_0)$ and let $N_{\mathrm{exp}}^\Delta$ be the discrete
jump operator defined in \eqref{eq:Q_def}.
Assume the exponential moment condition (Assumption~\ref{ass:laplace}) and define the Lipschitz constant:
\[
L_{\exp}(\delta)
:= \sup_{0\le i\le I} \lambda_i \left( \int_0^\infty e^{\delta x}\,\tilde\nu(\lambda_i,dx) + 1 \right) \;\le\; \lambda_{\max} \big(C_\delta^*(\Lambda) + 1\big)
<\infty.
\]
Then, for every pair of bounded grid functions $V,W$,
\begin{equation}\label{eq:CFL-explicit}
\| N_{\mathrm{exp}}^\Delta[V] - N_{\mathrm{exp}}^\Delta[W]\|_{\infty,\Lambda}
\;\le\; L_{\exp}(\delta)\,\|V-W\|_{\infty,\Lambda}.
\end{equation}
In particular, if the time step satisfies the explicit CFL condition 
\begin{equation}\label{eq:CFL-exp-cond}
\boxed{\Delta t\,L_{\exp}(\delta) \;\le\; c \;<\; 1}
\end{equation}
then $\Delta t\,N_{\mathrm{exp}}^\Delta$ is a contraction in $\ell^\infty$ with contraction
factor at most $c$.
\end{prop}

\begin{proof}
This is a direct application of Lemma~\ref{lem:Q-Lipschitz}, noting that the discrete
Gauss--Laguerre version inherits positivity of the weights and therefore the same
$\ell^\infty$–Lipschitz constant. Multiplying both sides of \eqref{eq:CFL-explicit} by
$\Delta t$ yields \eqref{eq:CFL-exp-cond}.
\end{proof}

\begin{remark}[Discrete version under Gauss--Laguerre + PCHIP]\label{rem:CFL-exp}
When $N_{\mathrm{exp}}^\Delta$ is implemented via the Gamma–mixture kernel, generalised
Gauss--Laguerre nodes $(x_{m,q},w_{m,q}^{(k_m-1)})$, and PCHIP interpolation, the constant
$L_{\exp}(\delta)$ in \eqref{eq:CFL-exp-cond} can be replaced by the explicit computable bound
\[
L_{\exp}^\Delta(\delta)
:= \max_{0\le i\le I}
\lambda_i\Bigg(
\sum_{m=1}^M \frac{\pi_m(\lambda_i)}{\Gamma(k_m)}
  \sum_{q=1}^Q w_{m,q}^{(k_m-1)}\,e^{\delta x_{m,q}}
\;+\;1\Bigg),
\]
so that a sufficient CFL condition is
\[
\Delta t \, L_{\exp}^\Delta(\delta) \;<\; 1
\]
\end{remark}

\begin{theorem}[One–step $L^\infty$–stability of the IMEX scheme]\label{thm:IMEX-stability}
Let $F^{n}$ and $G^{n}$ be two IMEX solutions of \eqref{eq:IMEX_step} corresponding to two
terminal data $F^{N},G^{N}$ at $t_N=T$. Suppose that the the explicit CFL-type \eqref{eq:CFL-exp-cond} holds. Then, for every
$n=0,\dots,N-1$,
\begin{equation}\label{eq:IMEX-stability}
\|F^{n} - G^{n}\|_{\infty,\Lambda}
\;\le\;
\frac{1+\Delta t\,L_{\exp}(\delta)}{1+\Delta t\,r}\;
\|F^{n+1} - G^{n+1}\|_{\infty,\Lambda}.
\end{equation}
In particular, iterating from $n=N-1$ down to $0$ gives the global stability bound
\[
\|F^{0} - G^{0}\|_{\infty,\Lambda}
\;\le\; \exp\!\big(L_{\exp}(\delta)\,T\big)\,\|F^{N} - G^{N}\|_{\infty,\Lambda}.
\]
\end{theorem}

\begin{proof}
Subtract the two IMEX steps :
\[
A(F^{n} - G^{n})
= (F^{n+1} - G^{n+1}) + \Delta t\bigl(
N_{\mathrm{exp}}^\Delta[F^{n+1}] - N_{\mathrm{exp}}^\Delta[G^{n+1}]
\bigr).
\]
Multiplying by $A^{-1}$, taking the supremum norm $\|\cdot\|_{\infty,\Lambda}$, and using the matrix norm compatibility $\|Mx\| \le \|M\|\|x\|$, we get:
\begin{align*}
\|F^{n} - G^{n}\|_{\infty,\Lambda} &\le \|A^{-1}\|_{\infty} \Bigl( \|F^{n+1} - G^{n+1}\|_{\infty,\Lambda} + \Delta t \|N_{\mathrm{exp}}^\Delta[F^{n+1}] - N_{\mathrm{exp}}^\Delta[G^{n+1}]\|_{\infty,\Lambda} \Bigr) \\
&\le \|A^{-1}\|_{\infty} \Bigl( 1 + \Delta t\,L_{\exp}(\delta) \Bigr) \|F^{n+1} - G^{n+1}\|_{\infty,\Lambda},
\end{align*}
using Proposition~\ref{prop:CFL-exp}.
Finally Proposition~\ref{prop:M-matrix} gives $\|A^{-1}\|_\infty\le (1+\Delta t r)^{-1}$,
which yields \eqref{eq:IMEX-stability}.
Iterating gives
\[
\|F^{0} - G^{0}\|_{\infty,\Lambda}
\le \prod_{n=0}^{N-1}\frac{1+\Delta t\,L_{\exp}(\delta)}{1+\Delta t\,r}\,
\|F^{N} - G^{N}\|_{\infty,\Lambda}
\le \exp\!\big(L_{\exp}(\delta)\,T\big)\,\|F^{N} - G^{N}\|_{\infty,\Lambda},
\]
since $(1+\Delta t\,L_{\exp})^N\le e^{L_{\exp}T}$ and $(1+\Delta t r)^{-N}\le 1$ for $r\ge0$.
\end{proof}

\begin{assumption}[Regularity and a-priori bounds for the modal PIDE]\label{ass:regularity}
Fix $\eta=\delta+iy$ with $\delta\in(0,s_0)$. On the truncated domain $\Lambda$,
assume that the modal PIDE \eqref{eq:PIDE_F_eta} admits a classical solution
$F(\cdot,\cdot,\eta)\in C^{2,2}([0,T]\times\Lambda)$.
Moreover, there exist constants $C,c>0$ (independent of $t$) such that for all $t\in[0,T]$,
\[
\|F(t,\cdot)\|_{\infty,\Lambda}
+\|\partial_t F(t,\cdot)\|_{\infty,\Lambda}
+\|\partial_\lambda F(t,\cdot)\|_{\infty,\Lambda}
+\|\partial_{tt}F(t,\cdot)\|_{\infty,\Lambda}
+\|\partial_{\lambda\lambda}F(t,\cdot)\|_{\infty,\Lambda}
\le C\,e^{c(T-t)}.
\]
Sufficient conditions ensuring these bounds (in terms of exponential moments of $\tilde\nu$ and weighted-TV regularity of $\lambda\mapsto \tilde\nu(\lambda,\cdot)$) are given in Appendix~\ref{app:regularity}.
\end{assumption}

\begin{theorem}[\bf IMEX consistency with PCHIP, Gauss--Laguerre and estimated measure]
\label{thm:consistency}
Fix $\eta=\delta+i y$ with $\delta\in(0,s_0)$ and consider the modal PIDE (\ref{eq:PIDE_F_eta}).

\noindent On the backward time grid $t_n=T-n\Delta t$ and spatial grid $\lambda_i=\lambda_{\min}+i h$ (where $h \equiv \Delta\lambda$),
the IMEX step reads
\begin{equation}\label{eq:scheme}
\frac{F_i^{n+1}-F_i^{n}}{\Delta t}
+\mu_i\,D_h^{\rm up}F^n(\lambda_i)
- r\,F_i^n
+\lambda_i\Big(\widehat Q_i^{\,n+1}-F_i^{n+1}\Big)
\;=\;0.
\end{equation}
where $D_h^{\rm up}$ is the \emph{implicit} upwind for $\mu\partial_\lambda$ and the explicit gain is evaluated by a $Q$-point
generalised Gauss--Laguerre rule under the \emph{estimated} measure $\tilde\nu(\lambda_i,dx)$ (Gamma mixture after Esscher),
with PCHIP interpolation $\Pi_{\Delta\lambda}$ and the calibrated parameter in GL $\widehat k_m, \widehat b_m, \widehat\pi_m$ and the weights/nodes $\widehat w_{m,q}, \widehat y_{m,q}$ :
\[
\widehat Q_i^{\,n+1}
=\sum_{m=1}^M \frac{\widehat \pi_m(\lambda_i)}{\Gamma(\widehat k_m)}
\sum_{q=1}^Q \widehat w_{m,q}^{(\widehat k_m-1)}\,
e^{\eta \widehat x_{m,q}}\,
\Pi_{\Delta\lambda}[F^{\,n+1}]\!\big(\lambda_i+\beta \widehat x_{m,q}\big),
\qquad \widehat x_{m,q}=\frac{\widehat z_{m,q}^{(\widehat k_m-1)}}{\widehat b_m}.
\]

\emph{Assumptions :}
(H1) $\exists s_0>0$ such that
$\int_0^\infty e^{s x}\,\tilde\nu(\lambda,dx)<\infty$ for $s\in[0,s_0)$ and all $\lambda$;\\
(H2) $\mu\in C^2$ on $[\lambda_{\min},\lambda_{\max}]$, with at-most-linear growth and bounded derivatives;\\
(H3) $\nu$ is a $\lambda$-measurable Gamma mixture (component densities and weights locally bounded);
(H4) truncated domain $\lambda\in \Lambda$;\\
(H5) $Q$-point GL quadrature; (H6) explicit CFL for the gain (see Remark~\ref{rem:CFL-exp}). 
\medskip

\noindent Define the local truncation residual by injecting the \emph{exact} solution $F$ 
into \eqref{eq:scheme}:
\[
\tau_i^n:=
\frac{F(t_{n+1},\lambda_i)-F(t_n,\lambda_i)}{\Delta t}
+\mu_i\,D_h^{\rm up}\!\big[F(t_n,\cdot)\big](\lambda_i)
- r\,F(t_n,\lambda_i)
+\lambda_i\Big(\widehat Q_i^{\,n+1}\big|_{F\ \text{injected}}-F(t_{n+1},\lambda_i)\Big).
\]
Then there exists a constant $C<\infty$ independent of $\Delta t,h,Q,i,n$ such that
\begin{equation}\label{eq:tau-bound}
|\tau_i^n|
\;\le\;
C\Big(
\underbrace{\Delta t}_{\text{backward Euler}}
+\underbrace{h}_{\text{upwind drift}}
+\underbrace{h^2}_{\text{PCHIP interp.}}
+\underbrace{\varepsilon_{\rm GL}(Q)}_{\text{GL quadrature}}
+\underbrace{\varepsilon_{\rm EM}(i)}_{\text{measure}}
+\underbrace{\varepsilon_{\rm bord}(i)}_{\text{boundary extrap.}}
\Big).
\end{equation}
In particular, if $\Delta t\to0$, $h\to0$, $Q\to\infty$, the estimated measure converges to $\tilde\nu$
(or is kept fixed sufficiently close), and $\varepsilon_{\rm bord}\to0$ (large enough domain),
then $\max_{i,n}|\tau_i^n|\to0$: the scheme \eqref{eq:scheme} is \emph{consistent}.
\end{theorem}

\begin{proof}[Proof (term-by-term decomposition and detailed bounds)]
Let
\(
Q_i^\star(t):=\int_0^\infty e^{\eta x}F\big(t,\lambda_i+\beta x\big)\,\tilde\nu(\lambda_i,dx)
\)
be the exact non-local gain under the underlying measure $\tilde\nu$.
Add/subtract appropriate terms to write
\begin{equation}\label{eq:tau-split}
\begin{aligned}
\tau_i^n
&=\underbrace{\Big(\frac{F(t_{n+1},\lambda_i)-F(t_n,\lambda_i)}{\Delta t}-\partial_t F(t_n,\lambda_i)\Big)}_{(i)\ \text{time discretization}}\
+\ \underbrace{\mu_i\Big(D_h^{\rm up}F(t_n,\lambda_i)-\partial_\lambda F(t_n,\lambda_i)\Big)}_{(ii)\ \text{space (implicit upwind)}}\\
&\quad+\underbrace{\lambda_i\Big( \big[Q_i^\star(t_{n+1};\tilde\nu)-F(t_{n+1},\lambda_i)\big] - \big[Q_i^\star(t_n;\tilde\nu)-F(t_n,\lambda_i)\big] \Big)}_{(iii)\ \text{explicit jump term shift}}\\
&\quad+\ \underbrace{\lambda_i\Big(Q_i^{\rm GL}(t_{n+1};\widehat\nu)-Q_i^\star(t_{n+1};\widehat\nu)\Big)}_{(iv)\ \text{GL quadrature}}\
+\underbrace{\lambda_i\Big(Q_i^{\rm GL,PCHIP}(t_{n+1};\widehat\nu)-Q_i^{\rm GL}(t_{n+1};\widehat\nu)\Big)}_{(v)\ \text{PCHIP interpolation}}\\
&\quad+\underbrace{\lambda_i\Big(Q_i^\star(t_{n+1};\widehat\nu)-Q_i^\star(t_{n+1};\tilde\nu)\Big)}_{(vi)\ \text{measure approximation}}\
+\ \underbrace{\lambda_i\,\mathcal E_{\rm bord}(i)}_{(vii)\ \text{boundary extrapolation}}.
\end{aligned}
\end{equation}

\smallskip
\noindent\textbf{(i) Time term.}
By Taylor with Cauchy reminder (see Assumption~\ref{ass:regularity} for the bounds of F),
\smallskip
$F(t_{n+1},\lambda_i)=F(t_n,\lambda_i)+\Delta t\,\partial_t F(t_n,\lambda_i)
+\tfrac12\Delta t^2\,\partial_{tt}F(t_n+\theta\Delta t,\lambda_i)$
for some $\theta\in(0,1)$, hence
\[
\big|(i)\big| \le \tfrac12\,\Delta t\,\|\partial_{tt}F\|_{\infty,\Lambda} \;=\; O(\Delta t).
\]

\smallskip
\noindent\textbf{(ii) Space (implicit upwind).}
Write the one–sided differences at $\lambda_i$:
\[
\frac{F(\lambda_{i+1})-F(\lambda_i)}{h}=\partial_\lambda F(\lambda_i)
+\tfrac{h}{2}\,\partial_{\lambda\lambda}F(\lambda_i+\theta_i^+ h),
\quad
\frac{F(\lambda_i)-F(\lambda_{i-1})}{h}=\partial_\lambda F(\lambda_i)
-\tfrac{h}{2}\,\partial_{\lambda\lambda}F(\lambda_i-\theta_i^- h),
\]
with $\theta_i^\pm\in(0,1)$.
For an upwind drift $\mu(\lambda_i)$,
\[
(D_h^{\rm up}F)(\lambda_i)=
\begin{cases}
\dfrac{F(\lambda_i)-F(\lambda_{i-1})}{h}, & \mu_i>0,\\[4pt]
\dfrac{F(\lambda_{i+1})-F(\lambda_i)}{h}, & \mu_i<0,
\end{cases}
\]
so in both cases
\(
\big|(D_h^{\rm up}F)-\partial_\lambda F\big|\le \frac{h}{2}\,\|\partial_{\lambda\lambda}F\|_{\infty,\Lambda}
\).
Therefore, by Assumption \ref{ass:regularity}
\[
\big|(ii)\big|\le |\mu_i|\,\frac{h}{2}\,\|\partial_{\lambda\lambda}F\|_{\infty,\Lambda}
\;=\; O(h).
\]

\smallskip
\noindent\textbf{(iii) Explicit time shift in the jump operator.}
Let $V(t) := Q_i^\star(t;\tilde\nu) - F(t,\lambda_i)$. By (H1), we may differentiate under the integral using dominated convergence (with bounding envelope $e^{\delta x}\|\partial_t F\|_\infty$). Thus:
\[
V'(t) = \int_0^\infty e^{\eta x} \partial_t F(t,\lambda_i+\beta x) \,\tilde\nu(\lambda_i,dx) - \partial_t F(t,\lambda_i).
\]
Taking the absolute value and bounding $|e^{\eta x}|$ by $e^{\delta x}$, we obtain $|V'(t)| \le \|\partial_t F(t,\cdot)\|_{\infty,\Lambda} \big(C_\delta^*(\Lambda) + 1\big)$. By the Mean Value Theorem,
\[
\big|(iii)\big| = \lambda_i \big| V(t_{n+1}) - V(t_n) \big| \le \lambda_{\max} \Delta t \sup_{s\in[t_n,t_{n+1}]} |V'(s)| \le \Delta t \lambda_{\max} \big(C_\delta^*(\Lambda) + 1\big) \sup_{s\in[t_n,t_{n+1}]}\|\partial_t F(s,\cdot)\|_{\infty,\Lambda}.
\]
Hence $\big|(iii)\big| = \mathcal{O}(\Delta t)$.

\smallskip
\noindent\textbf{(iv) Gauss--Laguerre quadrature error.}
Fix $m$ and set $\alpha=\widehat k_m-1>-1$. With the change of variables $x=y/\widehat b_m$,
write the $m$-component gain as
\[
I_{m,i}:=\frac{1}{\Gamma(\widehat k_m)}\int_0^\infty y^{\alpha}e^{-y}\,g_{m,i}(y)\,dy,
\qquad
g_{m,i}(y):=e^{\eta y/\widehat b_m}F\!\left(t_{n+1},\lambda_i+\beta\tfrac{y}{\widehat b_m}\right).
\]
Let $Y_{\max}=Y_{\max}(i,m)$ be such that $\lambda_i+\beta y/\widehat b_m\in \Lambda$ for all $y\in[0,Y_{\max}]$,
and split $I_{m,i}=I_{m,i}^{\rm in}+I_{m,i}^{\rm out}$ with
\[
I_{m,i}^{\rm in}:=\frac{1}{\Gamma(\widehat k_m)}\int_0^{Y_{\max}} y^{\alpha}e^{-y}\,g_{m,i}(y)\,dy,\qquad
I_{m,i}^{\rm out}:=\frac{1}{\Gamma(\widehat k_m)}\int_{Y_{\max}}^\infty y^{\alpha}e^{-y}\,g_{m,i}(y)\,dy.
\]
On $[0,Y_{\max}]$, the $Q$-point Gauss--Laguerre remainder yields
\[
\varepsilon_{\rm GL}(Q)\;\le\; \frac{R_Q(\alpha)}{\Gamma(\widehat k_m)}
\sup_{y\in[0,Y_{\max}]}\big|g_{m,i}^{(2Q)}(y)\big|,
\qquad
R_Q(\alpha)=\frac{\Gamma(Q+\alpha+1)\,Q!}{(2Q)!}=\mathcal O\!\big(4^{-Q}Q^{\alpha+1/2}\big),
\]
see \cite[Ch.~2 and Ch.~3.6]{DavisRabinowitz1984} or \cite[Ch.~3]{Gautschi2004}.

\smallskip
\noindent\textbf{(v) PCHIP interpolation inside the GL rule}
By linearity/ positivity of the weights,
\[
\big|Q_i^{\rm GL}-Q_i^{\rm GL,PCHIP}\big|
\le \sum_{m,q}\frac{\widehat \pi_m(\lambda_i)}{\Gamma(\widehat k_m)}\,
\widehat w_{m,q}^{(\widehat k_m-1)}\,e^{\delta \widehat x_{m,q}}\,
\|F(t_{n+1},\cdot)-\widetilde F^{\,n+1}\|_{\infty,\Lambda}
=: C_Q(i)\,\|F-\widetilde F\|_{\infty,\Lambda}.
\]
with $C_Q(i):=\sum_{m,q}\frac{\widehat \pi_m(\lambda_i)}{\Gamma(\widehat k_m)}\,
\widehat w_{m,q}^{(\widehat k_m-1)}\,e^{\delta \widehat x_{m,q}}$.
Using Lemma~\ref{lem:PCHIP_error}, $\|F-\widetilde F\|_{\infty,\Lambda}\le C_{\rm PCHIP}h^2$, and
noting that $\lambda_i\,C_Q(i)\le L_{\exp}^\Delta(\delta)$ (cf.\ Remark~\ref{rem:CFL-exp}), we get
$\big|(v)\big|=\lambda_i\big|Q_i^{\rm GL}-Q_i^{\rm GL,PCHIP}\big|
\le L_{\exp}^\Delta(\delta)\,C_{\rm PCHIP}\,h^2=O(h^2)$.

\smallskip
\noindent\textbf{(vi) Estimated measure (EM/Esscher)}
By linearity of the GL rule and $|e^{\eta x}|\le e^{\delta x}$,
\[
\big|Q_i^\star(t_{n+1};\tilde\nu)-Q_i^\star(t_{n+1};\widehat\nu)\big|
\le \| F^{\,n+1}\|_{\infty,\Lambda} \, \|\tilde\nu(\lambda_i,\cdot)-\widehat\nu(\lambda_i,\cdot)\|_{TV,\delta}\
=:\varepsilon_{\rm EM}(i).
\]
For Gamma mixtures under Esscher tilting (parameter $\theta$ so $b\mapsto b-\theta$; for the existence of the MGF $\delta<\min_m (b_m-\theta)$ on the compact set), Appendix~\ref{app:MoE}
proves the explicit bound
\[
\|\tilde\nu(\lambda_i)-\widehat\nu_{\theta}(\lambda_i)\|_{TV,\delta}\ \le\
\sum_{m=1}^M \big|\pi_m(\lambda)-\widehat\pi_m(\lambda)\big|\,M_{k_m,b_m}(\delta)
+\sum_{m=1}^M \widehat\pi_m(\lambda)\big(C_k|k_m-\widehat k_m|+C_b|b_m-\widehat b_m|\big),
\]
with $M_{k,b}(\delta)=(\tfrac{b}{b-\delta})^k$ and closed-form constants $C_k,C_b$ (Appendix~\ref{app:MoE}).
Thus $\varepsilon_{\rm EM}(i)$ is controlled by the softmax-weight error and the $(k,b)$ parameter error,
both vanishing under standard EM/MLE consistency on compact sets. See \cite{vanderVaart1998,White1982,NeweyMcFadden1994} for the asymptotics statistics and M-Estimation, \cite{McLachlanPeel2000} for the finite mixture model and the EM, and \cite{JacobsJordanNowlanHinton1991} for the MoE.

\smallskip
\noindent\textbf{(vii) Boundary/extrapolation error.}
If a GL node maps to $\lambda'_{i,m,q}=\lambda_i+\beta \widehat x_{m,q}\notin[\lambda_{\min},\lambda_{\max}]$,
we use a short linear extrapolation based on the PCHIP boundary slope $m_b$ (Fritsch--Carlson \cite{FritschCarlson1980} boundary formula).
Let $d>0$ be the signed distance from $\lambda_b\in\{\lambda_{\min},\lambda_{\max}\}$ to $\lambda'_{i,m,q}$.
By Taylor at the boundary,
\[
F(\lambda_b\pm d)
=F(\lambda_b)\pm d\,F'(\lambda_b) + \tfrac12 d^2\,F''(\xi_{b,d}),
\qquad \xi_{b,d}\in(\lambda_b,\,\lambda_b\pm d).
\]
The extrapolate is $E[\widetilde F](\lambda_b\pm d)=\widetilde F(\lambda_b)\pm d\,m_b$,
and since $\widetilde F(\lambda_b)=F(\lambda_b)$ (interpolation property),
\[
|E[\widetilde F](\lambda_b\pm d)-F(\lambda_b\pm d)|
\le d\,|m_b-F'(\lambda_b)| + \tfrac12 d^2\,\|F''\|_{\infty,\Lambda}.
\]
Appendix~\ref{app:PCHIP-boundary} proves $|m_b-F'(\lambda_b)|\le C h$ (both monotone and non-monotone cases).
Hence
\(
|\mathcal E_{\rm bord}(i)|\le C_Q(i)\big(C\,d_{\max}(i)\,h+\tfrac12 d_{\max}(i)^2\|F''\|_{\infty,\Lambda}\big)
=: \varepsilon_{\rm bord}(i).
\)
Choosing a domain large enough to keep $d_{\max}$ small (or zero) makes $\varepsilon_{\rm bord}$ negligible. 

\smallskip
\noindent Collecting (i)–(vii) and using that $\lambda_i$ is bounded on the compact domain, we obtain \eqref{eq:tau-bound}.
\end{proof}
\smallskip

\begin{theorem}[Convergence of the IMEX scheme]\label{thm:convergence}
Fix $\eta=\delta+iy$ with $\delta\in(0,s_0)$ and let $F(\cdot,\cdot,\eta)$ be a classical solution of
\eqref{eq:PIDE_F_eta} on $[0,T]\times\Lambda$, with $F(T,\lambda,\eta)\equiv 1$.
Let $\{F^n\}_{n=0}^N$ be produced by the discrete IMEX step \eqref{eq:IMEX_step}
(with $A=I-\Delta t\,L^\Delta_{\rm imp}$ and explicit gain $N^\Delta_{\exp}$), and impose the terminal condition exactly: $F^N_i=1$.

Assume the explicit CFL conditions \eqref{eq:CFL-exp-cond}.
Let $\tau^n$ be the local truncation residual obtained by injecting the exact solution into the scheme, and assume the uniform bound
\[
\max_{0\le n\le N-1}\|\tau^n\|_{\infty,\Lambda}\ \le\ \bar\tau,
\]
with $\bar\tau = C\big(\Delta t+h+h^2+\varepsilon_{\rm GL}+\varepsilon_{\rm EM}+\varepsilon_{\rm bord}\big)$ as in Thm.~\ref{thm:consistency}.
Then there exists a constant $C_T<\infty$, depending only on $T$, $r$ and $L_{\exp}(\delta)$, such that
for all $n=0,\dots,N$,
\[
\|F^n - F(t_n,\cdot,\eta)\|_{\infty,\Lambda}
\ \le\ C_T\,\bar\tau.
\]
In particular, the scheme is convergent as $\Delta t\to0$, $h\to0$, $Q\to\infty$, $\varepsilon_{\rm EM}\to0$ and $\varepsilon_{\rm bord}\to0$.
\end{theorem}

\begin{proof}
Let $e^n := F^n - F(t_n,\cdot,\eta)$. Subtracting the injected exact-solution identity from the numerical step yields
\[
A e^n
=
e^{n+1}+\Delta t\Big(N^\Delta_{\exp}[F^{n+1}] - N^\Delta_{\exp}[F(t_{n+1},\cdot,\eta)]\Big)
+\Delta t\,\tau^n.
\]
Multiplying by $A^{-1}$, taking the supremum norm $\|\cdot\|_{\infty,\Lambda}$, and using the Lipschitz bound of the explicit operator (Proposition~\ref{prop:CFL-exp}) gives:
\[
\|e^n\|_{\infty,\Lambda} \le \|A^{-1}\|_{\infty} \Big( \big(1 + \Delta t\,L_{\exp}(\delta)\big) \|e^{n+1}\|_{\infty,\Lambda} + \Delta t\,\|\tau^n\|_{\infty,\Lambda} \Big).
\]
Applying the Varah bound $\|A^{-1}\|_\infty \le \frac{1}{1+\Delta t\,r} \le 1$ from Proposition~\ref{prop:M-matrix} yields the one-step error recurrence:
\[
\|e^n\|_{\infty,\Lambda}\ \le\ K\,\|e^{n+1}\|_{\infty,\Lambda} + \Delta t\,\|\tau^n\|_{\infty,\Lambda},
\qquad \text{where} \quad
K:=\frac{1+\Delta t\,L_{\exp}(\delta)}{1+\Delta t\,r}.
\]
Iterating this recurrence backward from $n=N-1$ down to $0$, and noting that the terminal error is exactly zero ($e^N \equiv 0$), we obtain:
\[
\|e^n\|_{\infty,\Lambda}\ \le\ \Delta t \sum_{k=n}^{N-1} K^{k-n}\|\tau^k\|_{\infty,\Lambda}
\ \le\ \Delta t\,\bar\tau \sum_{j=0}^{N-1} K^j.
\]
Finally, as shown in the proof of Theorem~\ref{thm:IMEX-stability}, $K \le \exp(\Delta t\,L_{\exp}(\delta))$. Therefore, for any $j \le N$, $K^j \le \exp(j\Delta t\,L_{\exp}(\delta)) \le \exp(T\,L_{\exp}(\delta))$.
Bounding each term in the sum by this maximum exponential growth gives:
\[
\|e^n\|_{\infty,\Lambda}\ \le\ \Delta t\,\bar\tau\,N \exp\big(T\,L_{\exp}(\delta)\big) \ =\ T \exp\big(T\,L_{\exp}(\delta)\big)\,\bar\tau \ =:\ C_T\,\bar\tau.
\]
This concludes the proof.
\end{proof}

\medskip

\section{Reconstruction of the solution and numerical results} \label{sec:reconstruction-error}

We reconstruct the option value from the modal solutions
$F(t,\lambda,\delta+i y)$ using the Bromwich/Fourier inversion formula of
Theorem~\ref{thm:inversion} (see also \eqref{eq:V_reconstruction}).
Numerically, the $y$--integral is truncated to a finite window and discretised on a symmetric grid \cite{CarrMadan1999,Lee2004,TrefethenWeideman2014}.

\subsection{Discrete reconstruction}\label{sec:discrete-reconstruction}

Fix $Y_{\max}>0$, $\Delta y>0$, $K:=Y_{\max}/\Delta y\in\mathbb N$, and define
$y_k:=k\Delta y$ for $-K\le k\le K$.
Given numerical approximations $F^{\mathrm{num}}(t_n,\lambda_i,\delta+i y_k)$ and
$\widehat f_\delta^{\mathrm{num}}(y_k)$, we set
\begin{equation}\label{eq:V-num}
V^{\mathrm{num}}(t_n,\lambda_i,u)
:=\frac{\Delta y\,e^{\delta u}}{2\pi}
\sum_{k=-K}^{K}
\widehat f_\delta^{\mathrm{num}}(y_k)\,F^{\mathrm{num}}(t_n,\lambda_i,\delta+i y_k)\,e^{i y_k u},
\qquad u\in[0,U_{\max}].
\end{equation}

Assume $f$ is real-valued and all coefficients and the kernel $\tilde\nu(\lambda,\cdot)$ are real.
Then $\widehat f_\delta(-y)=\overline{\widehat f_\delta(y)}$. Moreover, the complex conjugate
$\overline{F(t,\lambda,\delta+i y)}$ solves the same modal PIDE as $F(t,\lambda,\delta-i y)$ with the same
terminal condition; by uniqueness of the modal solution, $F(t,\lambda,\delta-i y)=\overline{F(t,\lambda,\delta+i y)}$.
Hence the exact inversion is real and \eqref{eq:V-num} can equivalently be written as
\[
V^{\mathrm{num}}
=\frac{\Delta y\,e^{\delta u}}{2\pi}\Big(
\widehat f_\delta^{\mathrm{num}}(0)\,F^{\mathrm{num}}(\delta)
+2\sum_{k=1}^{K}\Re\!\big\{\widehat f_\delta^{\mathrm{num}}(y_k)\,F^{\mathrm{num}}(\delta+i y_k)\,e^{i y_k u}\big\}\Big).
\]
Any residual imaginary part in $V^{\mathrm{num}}$ is therefore a discretisation/roundoff artefact
and provides a useful diagnostic.

\subsection{Global reconstruction error}\label{sec:global-reconstruction-error}

For fixed $(t_n,\lambda_i,u)$ define $g(y):=\widehat f_\delta(y)\,F(t_n,\lambda_i,\delta+i y)\,e^{i y u}$, and let
\[
V_Y(t_n,\lambda_i,u)
:=\frac{e^{\delta u}}{2\pi}\int_{-Y_{\max}}^{Y_{\max}} g(y)\,dy,
\]
\[
V_Y^{\mathrm{trap}}(t_n,\lambda_i,u)
:=\frac{\Delta y\,e^{\delta u}}{2\pi}
\left[
\frac12 g(y_{-K})
+\sum_{k=-K+1}^{K-1} g(y_k)
+\frac12 g(y_K)
\right],
\]
i.e.\ $V_Y$ is the truncated inversion and $V_Y^{\mathrm{trap}}$ is the trapezoidal rule applied to the
\emph{exact} integrand on $[-Y_{\max},Y_{\max}]$.

\begin{lemma}[Global reconstruction error]\label{lem:global-V-error}
Assume Assumption~\ref{ass:laplace} and the integrability hypotheses of Theorem~\ref{thm:inversion}.
Fix $Y_{\max}$ and $\Delta y$ as above. Then, for any $u\in[0,U_{\max}]$,
\begin{equation}\label{eq:global-V-error}
|V^{\mathrm{num}}(t_n,\lambda_i,u)-V(t_n,\lambda_i,u)|
\;\le\;
E_{\mathrm{tail}}(Y_{\max}) + E_{\mathrm{trap}}(\Delta y;Y_{\max})
+ C_{\mathrm{rec}}(Y_{\max})\big(E_{\mathrm{PIDE}}+E_{\widehat f}\big),
\end{equation}
where
\[
E_{\mathrm{PIDE}}:=\max_{|k|\le K}\big|F^{\mathrm{num}}(t_n,\lambda_i,\delta+i y_k)-F(t_n,\lambda_i,\delta+i y_k)\big|,
\qquad
E_{\widehat f}:=\max_{|k|\le K}\big|\widehat f_\delta^{\mathrm{num}}(y_k)-\widehat f_\delta(y_k)\big|,
\]
\[
E_{\mathrm{tail}}(Y_{\max})
:=\frac{e^{\delta u}}{2\pi}\int_{|y|>Y_{\max}}
\big|\widehat f_\delta(y)\,F(t_n,\lambda_i,\delta+i y)\big|\,dy
\;\xrightarrow[Y_{\max}\to\infty]{}\;0,
\]
and one may take $C_{\rm rec}(Y_{\max})
:=\frac{e^{\delta u}}{2\pi}(2Y_{\max}+\Delta y)\,\max\{M_F(Y_{\max}),M_{\widehat f}(Y_{\max})\}$, with 

\noindent $M_F(Y_{\max})
:=\sup_{\substack{t\in[0,T],\,\lambda\in\Lambda\\ |y|\le Y_{\max}}}
\big|F(t,\lambda,\delta+i y)\big|$ and $M_{\widehat f}(Y_{\max})
:=\sup_{|y|\le Y_{\max}}|\widehat f_\delta(y)|$.

\noindent Moreover, if $g\in C^2([-Y_{\max},Y_{\max}])$, the composite trapezoidal error satisfies ( \cite{DavisRabinowitz1984,PressTeukolskyVetterlingFlannery2007}) :
\begin{equation}\label{eq:trap-bound-C2}
E_{\mathrm{trap}}(\Delta y;Y_{\max})
:=|V_Y^{\mathrm{trap}}-V_Y|
\;\le\;\frac{e^{\delta u}}{2\pi}\cdot \frac{(2Y_{\max})}{12}\,(\Delta y)^2
\max_{|y|\le Y_{\max}}|g''(y)|.
\end{equation}
In particular, for fixed $Y_{\max}$, if $E_{\mathrm{PIDE}}\to0$, $E_{\widehat f}\to0$ and $\Delta y\to0$
then $V^{\mathrm{num}}\to V_Y$ uniformly on $u\in[0,U_{\max}]$; letting $Y_{\max}\to\infty$ then yields $V_Y\to V$.
\end{lemma}

\begin{proof}
Decompose
\[
|V^{\mathrm{num}}-V|
\le |V^{\mathrm{num}}-V_Y^{\mathrm{trap}}| + |V_Y^{\mathrm{trap}}-V_Y| + |V_Y-V|.
\]
The last term is $E_{\mathrm{tail}}(Y_{\max})$ by definition.

\noindent The middle term is $E_{\mathrm{trap}}$. Since $g\in C^2([-Y_{\max},Y_{\max}])$, the standard remainder
formula for the composite trapezoidal rule on a finite interval yields
\[
\left|\int_{-Y_{\max}}^{Y_{\max}} g(y)\,dy - \Delta y\sum_{k=-K}^{K} g(y_k)\right|
\le \frac{(2Y_{\max})}{12}(\Delta y)^2 \max_{|y|\le Y_{\max}}|g''(y)|,
\]
which gives \eqref{eq:trap-bound-C2} after multiplying by $e^{\delta u}/(2\pi)$. For the first term, at each node $y_k$ write
\[
\widehat f_\delta^{\mathrm{num}}(y_k)F^{\mathrm{num}}(y_k)-\widehat f_\delta(y_k)F(y_k)
=\big(\widehat f_\delta^{\mathrm{num}}-\widehat f_\delta\big)F^{\mathrm{num}}
+\widehat f_\delta\big(F^{\mathrm{num}}-F\big),
\]
take absolute values and sum over $|k|\le K$. Using $(2K+1)\Delta y = 2Y_{\max}+\Delta y$ yields
\[
|V^{\mathrm{num}}-V_Y^{\mathrm{num}}|
\le \frac{e^{\delta u}}{2\pi}(2Y_{\max}+\Delta y)
\Big( E_{\widehat f}\,\max_{|k|\le K}|F^{\mathrm{num}}(t_n,\lambda_i,\delta+i y_k)|
      + \max_{|k|\le K}|\widehat f_\delta(y_k)|\,E_{\mathrm{PIDE}}\Big),
\]
and the claimed form follows by absorbing the bounded factors into $C_{\mathrm{rec}}(Y_{\max})$.
\end{proof}

\begin{remark}[Sharper trapezoidal convergence]\label{rem:sharper-spectral}
The $O((\Delta y)^2)$ bound \eqref{eq:trap-bound-C2} is the generic $C^2$ finite-interval estimate.
In many transform pricing problems, $y\mapsto g(y)$ extends analytically to a strip around the real axis,
in which case the equispaced trapezoidal rule can converge exponentially fast in $1/\Delta y$; see
\cite{TrefethenWeideman2014} and the error-control discussion in \cite{Lee2004}.
\end{remark}

\subsection{Pricing algorithm (summary)}
\begin{enumerate}
\item Choose a damping $\delta>0$, a cut-off $Y_{\max}$ and a frequency step $\Delta y$.
\item Compute $\widehat f_\delta(y_k)$ on $y_k=k\Delta y$ (closed form / GL / FFT).
\item For each $\eta_k=\delta+i y_k$, solve \eqref{eq:PIDE_F_eta} backward in $(t,\lambda)$ with the IMEX--GL--PCHIP scheme to obtain $F^{\rm num}(t_n,\lambda_i,\eta_k)$.
\item Reconstruct $V^{\rm num}(t_n,\lambda_i,u)$ by \eqref{eq:V-num} (or the real symmetric form).
\end{enumerate}

\section{Numerical experiments}\label{sec:numerics}
 
We illustrate the IMEX--Gauss--Laguerre--Bromwich pipeline on a capped call payoff
$f(u)=\min\bigl((u-K)^+,C\bigr)$ written on the accumulated mark process $U_T$ at maturity $T$. We solve it as described in Section~\ref{sec:IMEX}: the local drift (OU transport) and discount terms are treated implicitly, while the nonlocal jump operator is treated explicitly. 
The model state is $(U_t,\lambda_t)$, where $\lambda_t$ is the jump intensity (in year$^{-1}$) and
$U_t=\sum_{k\le N_t}X_k$ accumulates jump magnitudes. Between jumps, $\lambda_t$ follows the
mean-reverting ODE $\dot{\lambda}_t=\kappa(\bar{\lambda}-\lambda_t)$, and at each jump of size $X>0$,
we update $U\leftarrow U+X$ and $\lambda\leftarrow \lambda+\beta X$. Hence $\bar\lambda$ is the long-run
baseline activity, $\lambda_0$ the initial activity level, $\kappa$ the mean-reversion speed and $\beta$
the self-excitation parameter.
\\
In the numerical experiments we implement $\Pi_{\Delta\lambda}$ in \eqref{eq:Q_discrete_GL} with piecewise-linear interpolation and use a two-component Gamma mixture ($M=2$) with constant weights $\pi_1=p_{\mathrm{mix}}$ and $\pi_2=1-p_{\mathrm{mix}}$. These choices simplify the implementation without changing the IMEX--Gauss--Laguerre--Bromwich structure. In particular, replacing PCHIP by piecewise-linear interpolation preserves monotonicity and $\ell^\infty$-stability
(1-Lipschitz), at the cost of reducing the interpolation order from $O(\Delta\lambda^2)$ to $O(\Delta\lambda)$.
Likewise, Simpson composite rule is used for the Bromwich integral in place of the trapezoidal rule in Section~\ref{sec:reconstruction-error} with $N_y$ points.

\smallskip
\noindent Jump magnitudes follow a two-component Gamma mixture ($M=2$),
$X \sim p_{\mathrm{mix}}\,\Gamma(k_1,b_1) + (1-p_{\mathrm{mix}})\,\Gamma(k_2,b_2)$,
with shape parameters $k_m>0$ and \emph{rate} parameters $b_m>0$ (so that $\mathbb{E}[X]=k_m/b_m$ for each component). In all experiments we choose the Bromwich shift $\delta$ such that $0<\delta<\min(b_1,b_2)$,
which guarantees finiteness of the moment generating function $\mathbb{E}[e^{(\delta+i y)X}]$ (cf. Assumption~\ref{ass:laplace}) and hence well-posedness of the GL approximation in \eqref{eq:Q_discrete_GL}.

\smallskip
\noindent We also verify the sufficient stability condition discussed in Remark \ref{rq:stability-cdtn}, ensuring that mean reversion dominates the average self-excitation from marks.
In our experiments the mark law does not depend on $\lambda$ (constant mixture weights), hence
$\sup_{\lambda\in\Lambda}m_1^\theta(\lambda)=m_1^\theta$ is constant. Under the actuarial kernel
($\theta=0$) we have $m_1 = p_{\mathrm{mix}}\frac{k_1}{b_1} + (1-p_{\mathrm{mix}})\frac{k_2}{b_2}
=0.6\cdot 0.5 + 0.4\cdot 2.4 = 1.26$ so that even at $\beta=2$ we have $\beta m_1 \approx 2.52 \ll \kappa=8$.
We monitor the quadrature-weighted boundary-hit ratio (fraction of GL mass with $\lambda_i+\beta x>\lambda_{\max}$).
In all baseline runs it remains below $1\%$, so the boundary term $\varepsilon_{\rm bord}$ in
Theorem~\ref{thm:consistency} is negligible for $\lambda_{\max}=450$.
\medskip

\begin{table}[t]
\centering
\caption{Baseline model parameters and numerical settings (used in Fig.~\ref{fig:price_vs_beta_publishable}
and the convergence diagnostics).}
\label{tab:baseline_beta}
\begin{tabular}{ll}
\hline
\textbf{Model} &
$T=\frac{150}{365}$,
$r=0.02$,
$\kappa=8$,
$\bar\lambda=2$,
$\lambda_0=2.5$,
$u_0=0$ \\
\textbf{Marks (mixture)} &
$p_{\mathrm{mix}}=0.6$,
$(k_1,b_1)=(2,4)$,
$(k_2,b_2)=(6,2.5)$ \\
\textbf{Payoff} &
$K=1.2$,
$C=3.0$ \\
\textbf{Self-excitation} &
$\beta=1.0$ (varied in Fig.~\ref{fig:price_vs_beta_publishable}) \\
\hline
\textbf{PIDE grid} &
$\lambda_{\max}=450$,
$N_\lambda=600$,
$\Delta t=\frac{1}{365}$,
$Q=24$ \\
\textbf{Bromwich} &
$\delta=0.3$,
$Y_{\max}=120$,
$N_y=1537$ (Simpson) \\
\hline
\end{tabular}
\end{table}

\subsection{Sensitivity to self-excitation: price versus $\beta$}\label{subsec:price_beta}

We first investigate the impact of self-excitation on option values. Recall that, after a jump of size $X>0$,
the intensity is updated as $\lambda \leftarrow \lambda + \beta X$, so that larger $\beta$ amplifies the feedback
from marks to future event rates. Since the capped-call payoff $f(U_T)$ is nondecreasing in $U_T$, stronger self-excitation increases the likelihood of additional jumps prior to maturity and
therefore raises the option price.

\begin{figure}[t]
\centering
\includegraphics[width=0.60\textwidth]{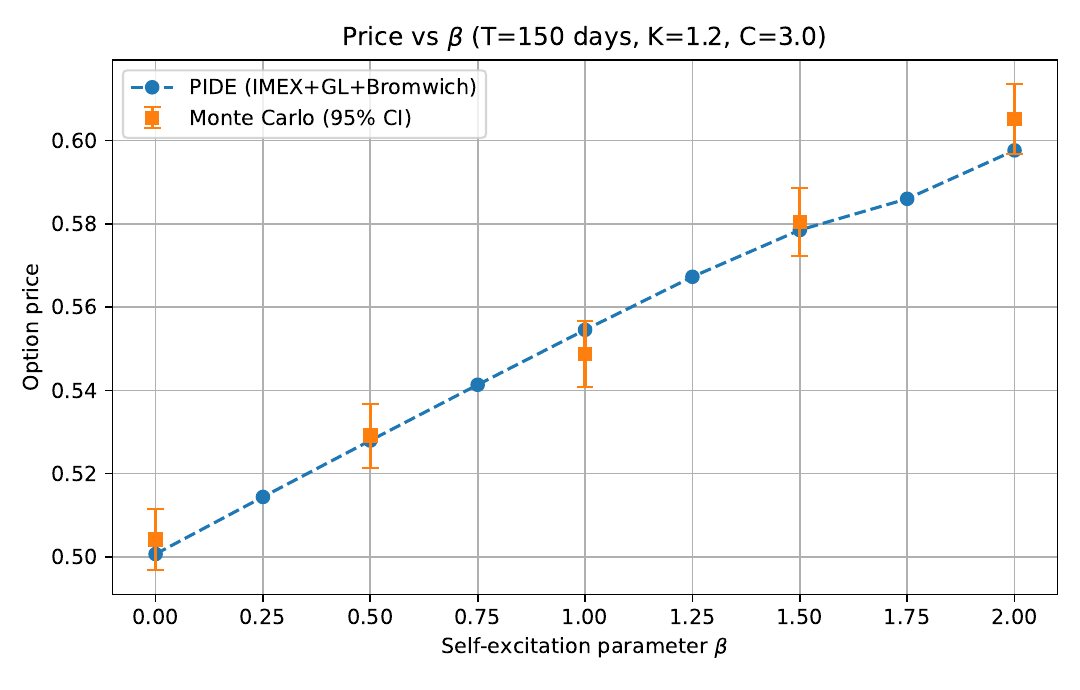}
\caption{Option price as a function of the self-excitation strength $\beta$ for $T=150$ days and $(K,C)=(1.2,3.0)$.
Dashed line with markers: IMEX--Gauss--Laguerre--Bromwich price.
Squares with $95\%$ error bars: Monte Carlo benchmark (exact PDMP simulation).
All remaining parameters are fixed.}
\label{fig:price_vs_beta_publishable}
\end{figure}

\noindent Figure~\ref{fig:price_vs_beta_publishable} reports the price as a function of $\beta$ for a representative contract
with maturity $T=150$ days and payoff parameters $(K,C)=(1.2,3.0)$, chosen so that the option is sufficiently
sensitive to changes in the distribution of $U_T$. The IMEX--Gauss--Laguerre--Bromwich prices are benchmarked against
Monte Carlo estimates (exact PDMP simulation) with $95\%$ confidence intervals. The agreement between both methods
confirms the correctness of the numerical pipeline, and the monotone increase in $\beta$ illustrates the economic role
of self-excitation in amplifying accumulated jump activity.

\subsection{Convergence diagnostics: time step and spatial grid}\label{subsec:conv_dt_space}

We assess numerical convergence by reporting \emph{absolute} errors with respect to a finer \emph{numerical}
reference value $V_{\mathrm{ref}}$, computed with a substantially refined discretisation in time, space and in the
Bromwich inversion (larger window and higher frequency resolution).
We stress that $V_{\mathrm{ref}}$ is \emph{not} an exact price (no closed form is available in general), but rather an
\emph{internal} high-accuracy benchmark. Consequently, the reported quantities
$|V(\cdot)-V_{\mathrm{ref}}|$ should be interpreted as \emph{convergence indicators} rather than absolute errors.

\paragraph{Time refinement.}
In Figure~\ref{fig:error_vs_dt}, we refine the time step $\Delta t$ while keeping all other parameters fixed, and
plot the error $|V(\Delta t)-V_{\mathrm{ref}}|$.
The error decreases steadily as $\Delta t\to 0$, displaying the first order $\mathcal{O}(\Delta t)$ asymptotic behavior. This is perfectly consistent with IMEX-Euler time stepping (implicit Euler for the local drift/discount part combined with an explicit treatment of the nonlocal jump operator).
This experiment confirms that, for the baseline contract, time discretisation error can be reduced predictably by
temporal refinement.

\begin{figure}[t]
\centering

\begin{subfigure}[t]{0.48\textwidth}
    \centering
    \includegraphics[width=\textwidth]{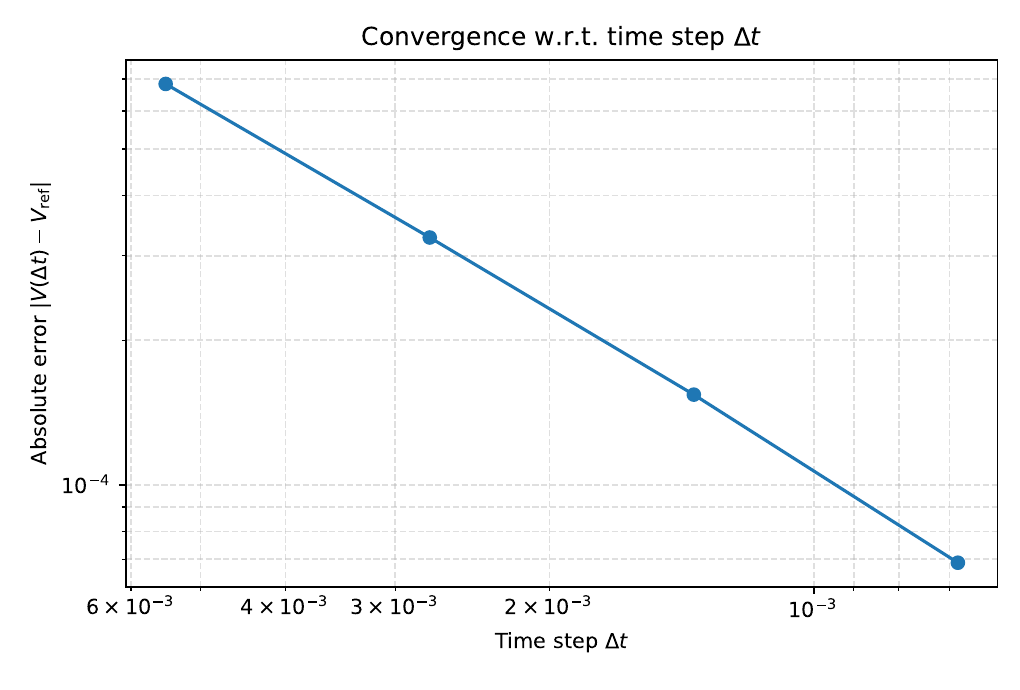}
    \caption{Convergence with respect to the time step $\Delta t$.}
    \label{fig:error_vs_dt}
\end{subfigure}
\hfill
\begin{subfigure}[t]{0.48\textwidth}
    \centering
    \includegraphics[width=\textwidth]{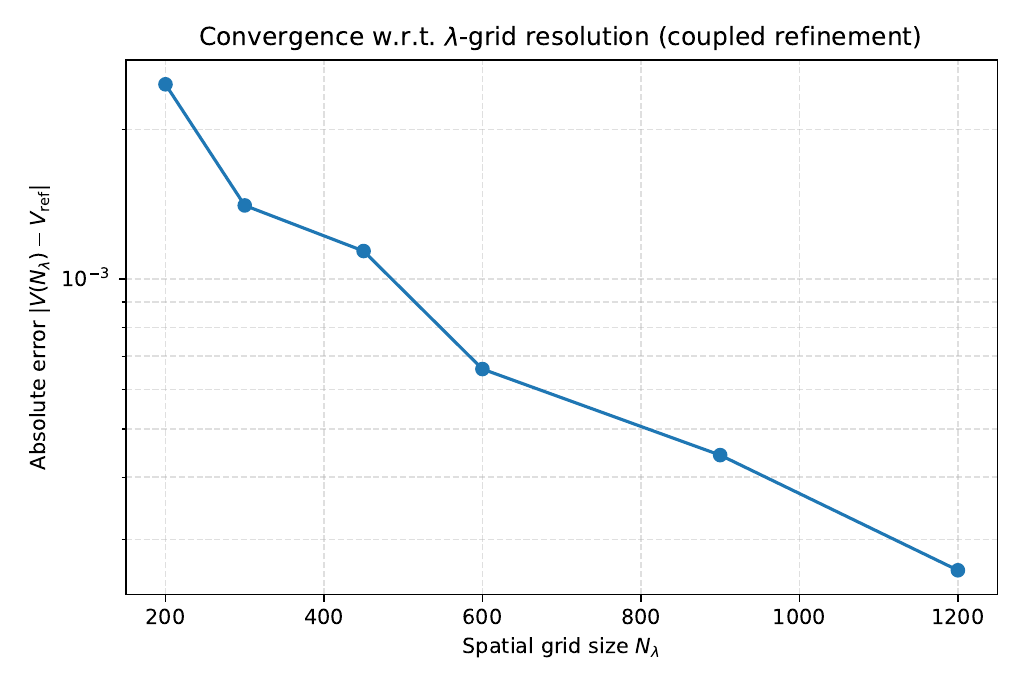}
    \caption{Convergence with respect to the $\lambda$-grid resolution.}
    \label{fig:error_vs_Nlam}
\end{subfigure}

\caption{Convergence diagnostics against a finer internal numerical reference $V_{\mathrm{ref}}$. Left: time refinement. Right: spatial refinement along a coupled path where $\Delta t$ decreases as $N_\lambda$ increases.}
\label{fig:conv_dt_space}
\end{figure}

\paragraph{Spatial refinement.}
In Figure~\ref{fig:error_vs_Nlam}, we study the impact of the $\lambda$-grid resolution at fixed $\lambda_{\max}$.
Since the scheme involves off-grid evaluations at shifted points $\lambda+\beta x$ (handled here by
piecewise-linear interpolation) and since the overall price also depends on the discretisation of the Bromwich
integral, varying $N_\lambda$ alone may lead to mixed error regimes.
To obtain a clean and interpretable trend, we therefore refine along a \emph{coupled} path: as $N_\lambda$
increases, we simultaneously decrease $\Delta t$ (and keep the Bromwich configuration fixed) so that the time
discretisation does not dominate the observed behaviour.
Along this refinement path, the error $|V(N_\lambda)-V_{\mathrm{ref}}|$ decreases with $N_\lambda$, indicating that the
spatial discretisation in $\lambda$ converges through the theoretical first order $\mathcal{O}(1/N_\lambda)=\mathcal{O}(\Delta\lambda)$ error and is not the limiting factor in the baseline configuration.

\subsection{Bromwich inversion diagnostics}\label{subsec:bromwich_diagnostics}

We next assess the numerical error introduced by the Bromwich inversion used to reconstruct the option price from the transformed solution. Since no closed form is available for the full IMEX--GL--Bromwich pipeline, we compare against an \emph{internal numerical reference} $V_{\mathrm{ref}}$ computed with a wider Bromwich window and a finer frequency grid, while keeping the PIDE discretisation fixed at a sufficiently fine level. Consequently, the reported absolute differences $|V(N_y)-V_{\mathrm{ref}}|$ should be interpreted as practical convergence indicators rather than exact errors.

\paragraph{Discretisation of the Bromwich integral.}
\noindent To largely decouple the discretisation error from the truncation error, we fix a large truncation level ($Y_{\max}=200$) and increase the number of frequency points $N_y$ in the composite Simpson rule.
Figure~\ref{fig:error_vs_Ny} reports the resulting absolute differences $|V(N_y)-V_{\mathrm{ref}}|$.
The error decreases very rapidly as $N_y$ increases, which is consistent with the high regularity (often analyticity) of the Bromwich integrand encountered in transform-based pricing.
In our experiments, $N_y$ in the range $2000$--$3000$ already yields errors below $10^{-4}$, and further refinement pushes the discretisation error into the $10^{-6}$ regime. This indicates that, for the baseline configuration, the frequency discretisation is not a numerical bottleneck compared to the PIDE resolution.

\noindent In contrast, varying $Y_{\max}$ at fixed step $\Delta y$ may lead to mild non-monotonic behaviour due to the
oscillatory nature of the truncated Bromwich integral; in our experiments the dominant effect is the frequency
discretisation, hence we report convergence in $N_y$.

\begin{figure}[H]
\centering
\includegraphics[width=0.60\textwidth]{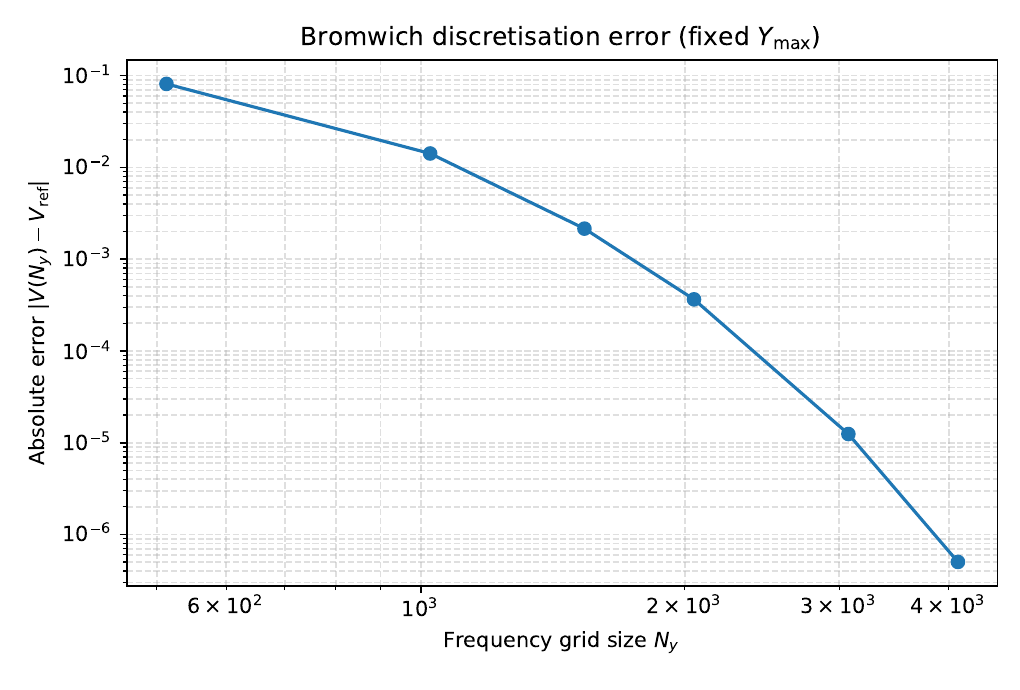}
\caption{Convergence of the Bromwich inversion with respect to the frequency grid size $N_y$ for a fixed truncation
level $Y_{\max}$. The error is measured against an internal numerical reference $V_{\mathrm{ref}}$ computed with a
larger $(Y_{\max},N_y)$.}
\label{fig:error_vs_Ny}
\end{figure}

\subsection{Sensitivity to the initial intensity $\lambda_0$}\label{subsec:greeks_lambda0}

We report the sensitivity of option values to the initial jump intensity $\lambda_0$, which
controls the short-term activity level of the marked point process. Since a larger $\lambda_0$
increases the expected number of jumps prior to maturity, it increases the accumulated mark
$U_T$ in distribution and therefore raises the capped-call payoff $f(U_T)=\min((U_T-K)^+,C)$.

\begin{figure}[H]
\centering
\includegraphics[width=0.60\textwidth]{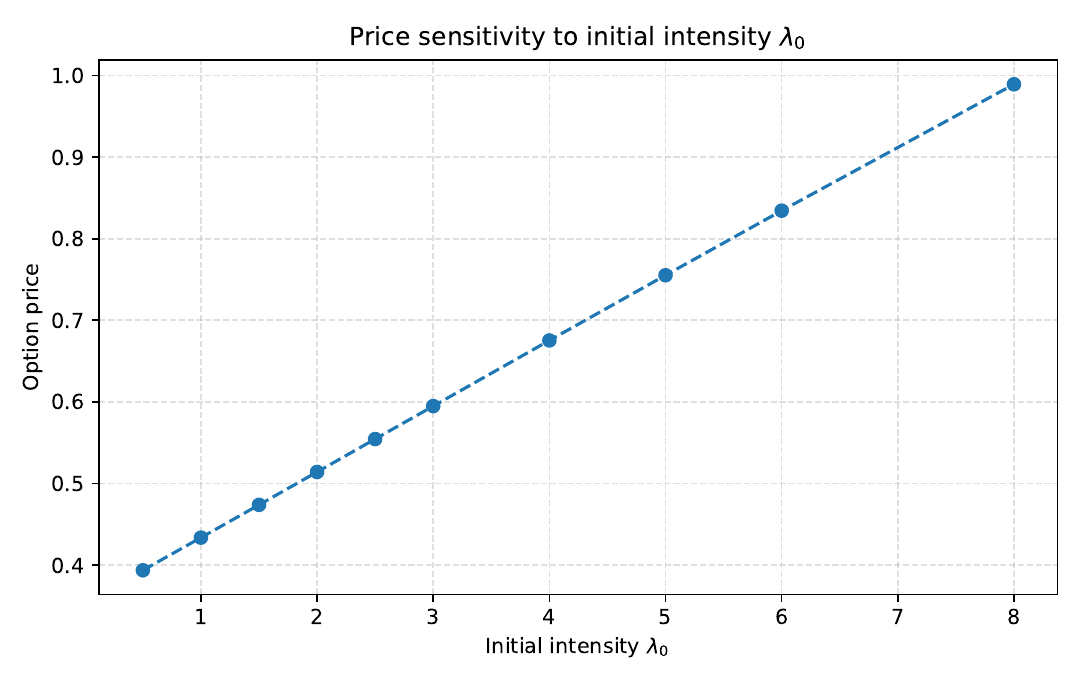}
\caption{Option price as a function of the initial jump intensity $\lambda_0$. The monotonic increase reflects the higher probability of accumulating large marks prior to maturity.}
\label{fig:price_vs_lambda0}
\end{figure}

\noindent Figure~\ref{fig:price_vs_lambda0} shows that the price is monotonically increasing in $\lambda_0$ over the
range considered. This behaviour is perfectly consistent with the economic intuition above.

\begin{figure}[H]
\centering
\includegraphics[width=0.60\textwidth]{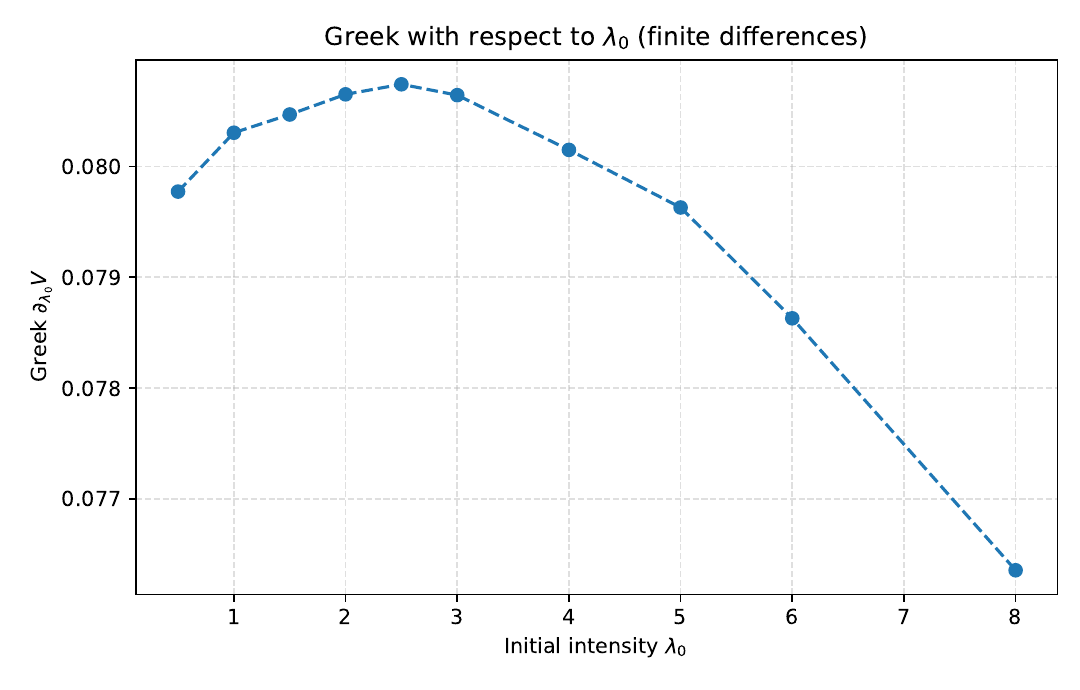}
\caption{The ``intensity Greek'' (Delta) $\partial_{\lambda_0}V$ estimated via centred finite differences. The smooth, bell-shaped profile highlights the diminishing marginal impact of extreme initial intensities (saturation effect).}
\label{fig:delta_lambda0}
\end{figure}

\noindent We further estimate the ``intensity Greek'' $\partial_{\lambda_0}V$ by centred finite differences,
\[
\partial_{\lambda_0}V(\lambda_0)\approx \frac{V(\lambda_0+h)-V(\lambda_0-h)}{2h},
\]
with a small relative step $h$. Figure~\ref{fig:delta_lambda0} shows that $\partial_{\lambda_0}V$
is positive and relatively flat (around $8\times 10^{-2}$ in our baseline). Notably, it exhibits a slight bell-shaped profile: it reaches a local maximum around $\lambda_0 \approx 2.5$ before mildly decreasing as $\lambda_0$ becomes larger. This captures the diminishing marginal impact of the initial intensity, which is expected due to the payoff cap $C$ (saturation effect) and the mean-reversion of $\lambda_t$.

\section{Conclusion}\label{sec:conclusion}

We developed a pricing framework for cumulative jump claims under a marked self-exciting intensity model with mean reversion. On the theoretical side, we derived the backward PIDE representation under both actuarial and Esscher-type pricing conventions and established consistency and stability properties for an IMEX discretisation combined with Gauss--Laguerre quadrature and Bromwich inversion. On the numerical side, the experiments confirm the accuracy and robustness of the pipeline: prices behave consistently with the economic role of self-excitation, the IMEX discretisation converges as expected in time and space, and the Bromwich inversion is not a numerical bottleneck once the frequency grid is sufficiently refined. 

Several extensions may be explored in future work, including state-dependent mixture weights, alternative interpolation and inversion procedures, and calibration to real precipitation data under statistical and pricing measures.

\newpage 
\clearpage
\appendix

\section{PCHIP interpolation error (details)}\label{app:PCHIP_error_detail}

Let $F\in C^2([\lambda_{\min},\lambda_{\max}])$ and let $\widetilde F=\Pi_{\Delta\lambda}[F]$
be the PCHIP (Fritsch--Carlson) interpolant on a grid with maximal step $h$.
Then there exists $C>0$ (depending on $\|F''\|_{\infty}$ and mesh regularity) such that
$\|\widetilde F-F\|_\infty\le Ch^2$.

On each cell $[\lambda_j,\lambda_{j+1}]$, write $\lambda=\lambda_j+\theta h_j,\qquad \theta\in[0,1],\qquad h_j=\lambda_{j+1}-\lambda_j$.
The cubic Hermite representation of the PCHIP interpolant is
\begin{equation}\label{eq:Hermite_form}
\widetilde F(\lambda)
= H_{00}(\theta)F_j + H_{10}(\theta)h_j m_j
+ H_{01}(\theta)F_{j+1} + H_{11}(\theta)h_j m_{j+1},
\end{equation}
where $H_{00},H_{10},H_{01},H_{11}$ are the standard cubic Hermite basis functions.

Next, Taylor expansion with Cauchy remainder gives
\[
F(\lambda)=F_j+\theta h_j\,s_j+\frac12\theta^2 h_j^2\,F''(\varepsilon_\theta),
\qquad
F_{j+1}=F_j+h_j\,s_j+\frac12 h_j^2\,F''(\varepsilon_\beta),
\]
for some $\varepsilon_\theta,\varepsilon_\beta\in(\lambda_j,\lambda_{j+1})$, where $s_j:=F'(\lambda_j)$.

Substituting these expansions into \eqref{eq:Hermite_form}, and using $H_{00}(\theta)+H_{01}(\theta)=1$, $\theta-H_{01}(\theta)=H_{10}(\theta)+H_{11}(\theta)$,
we obtain
\[
F(\lambda)-\widetilde F(\lambda)
=
h_j\Big(H_{10}(\theta)\big(s_j-m_j\big)+H_{11}(\theta)\big(s_j-m_{j+1}\big)\Big)
+\frac12 h_j^2\Big(\theta^2F''(\varepsilon_\theta)-H_{01}(\theta)F''(\varepsilon_\beta)\Big).
\]
Hence
\[
F(\lambda)-\widetilde F(\lambda)
= \mathcal O(h_j^2)\|F''\|_\infty
+ h_j\big(|F'(\lambda_j)-m_j|+|F'(\lambda_{j+1})-m_{j+1}|\big),
\]
where $m_j$ are the PCHIP nodal slopes.
It suffices to show the first-order slope accuracy $|m_j-F'(\lambda_j)|=\mathcal O(h)$.
If $d_{j-1}d_j>0$ (locally monotone data), $m_j$ is the weighted harmonic mean of secants
$d_{j-1},d_j$ and Taylor expansion gives $d_{j\pm1}=F'(\lambda_j)+\mathcal O(h)$, hence
$m_j=F'(\lambda_j)+\mathcal O(h)$; see Fritsch--Carlson~\cite{FritschCarlson1980} (and Hyman~\cite{Hyman1983}).
If $d_{j-1}d_j\le 0$ (non-monotone), PCHIP sets $m_j=0$ and the mean value theorem implies the
existence of $\gamma\in(\lambda_{j-1},\lambda_{j+1})$ with $F'(\gamma)=0$, so
$|F'(\lambda_j)|\le \|F''\|_\infty|\lambda_j-\gamma|=\mathcal O(h)$, hence again $|m_j-F'(\lambda_j)|=\mathcal O(h)$.
Combining gives $\|F-\widetilde F\|_\infty=\mathcal O(h^2)$.

\section{Detailed a-priori bounds and $C^{1,2}$ regularity}\label{app:regularity}

Fix $\eta=\delta+iy$ with $\delta\in(0,s_0)$ and work on $\Lambda=[\lambda_{\min},\lambda_{\max}]$.
Consider the modal PIDE
\[
\partial_t F + \mu(\lambda)\partial_\lambda F + \lambda\,\mathcal J_\eta[F] - rF = 0,
\qquad F(T,\lambda,\eta)=1,
\]
with
\[
\mathcal J_\eta[F](t,\lambda):=\int_0^\infty\Big(e^{\eta x}F(t,\lambda+\beta x,\eta)-F(t,\lambda,\eta)\Big)\tilde\nu(\lambda,dx).
\]
The existence of a classical solution is assumed; the bounds below follow from standard
characteristic estimates and weighted-TV continuity of the jump kernel, in line with the PIDE
literature (see e.g. \cite{ContTankov2004,BarlesImbert2008,CrandallIshiiLions1992}).
\begin{assumption}\label{ass:reg_suff}
(i) (\textbf{Exponential moment}) $M_0:=\sup_{\lambda\in\Lambda}\int_0^\infty e^{\delta x}\tilde\nu(\lambda,dx)<\infty$.

(ii) (\textbf{Kernel regularity}) $\lambda\mapsto \tilde\nu(\lambda,\cdot)$ is differentiable in weighted TV and \newline
$\sup_{\lambda\in\Lambda}\|\partial_\lambda \tilde\nu(\lambda,\cdot)\|_{TV,\delta}\le L_\nu<\infty$.

(iii) $\mu\in C^1(\Lambda)$ with $\|\mu'\|_{\infty,\Lambda}<\infty$ and $\lambda\in\Lambda \Rightarrow 0\le \lambda\le \lambda_{\max}$.
\end{assumption}

\begin{lemma}[Derivative of the jump operator]\label{lem:J_eta_derivative}
Under Assumption~\ref{ass:reg_suff}, for $F\in C^1(\Lambda)$,
\[
\|\partial_\lambda \mathcal J_\eta[F](t,\cdot)\|_{\infty,\Lambda}
\le 2M_0\|F_\lambda(t,\cdot)\|_{\infty,\Lambda}+2L_\nu\|F(t,\cdot)\|_{\infty,\Lambda}.
\]
\end{lemma}

\begin{proof}[Proof sketch]
Write $\mathcal J_\eta[F](\lambda)=\int e^{\eta x}\psi(\lambda,x)\tilde\nu(\lambda,dx)$ with
$\psi(\lambda,x)=F(\lambda+\beta x)-F(\lambda)$ and use difference quotients:
the term where $\tilde\nu$ is frozen is bounded by $2M_0\|F_\lambda\|_\infty$ by MVT, and the term
where $\psi$ is frozen is controlled by $2\|F\|_\infty \|\partial_\lambda \tilde\nu\|_{TV,\delta}$.
\end{proof}

\begin{prop}[A-priori bounds]\label{prop:apriori_bounds}
Assume $F(\cdot,\cdot,\eta)$ is a classical solution. Under Assumption~\ref{ass:reg_suff},
there exist constants $C,c>0$ such that for all $t\in[0,T]$,
\[
\|F(t,\cdot)\|_{\infty,\Lambda}+\|F_\lambda(t,\cdot)\|_{\infty,\Lambda}+\|F_t(t,\cdot)\|_{\infty,\Lambda}
\le C e^{c(T-t)}.
\]
If moreover $\sup_{\lambda\in\Lambda}\|\partial_{\lambda\lambda}\tilde\nu(\lambda,\cdot)\|_{TV,\delta}<\infty$
and $\mu\in C^2(\Lambda)$, analogous bounds hold for $\|F_{\lambda\lambda}\|_\infty$ and $\|F_{tt}\|_\infty$.
\end{prop}

\begin{proof}[Proof sketch]
Along characteristics $\dot\Lambda=\mu(\Lambda)$, the PIDE gives
$\frac{d}{ds}F(s,\Lambda(s))=rF-\Lambda(s)\mathcal J_\eta[F]$.
Using $|e^{\eta x}|\le e^{\delta x}$ and Assumption~\ref{ass:reg_suff}(i) yields
$|\mathcal J_\eta[F]|\le (M_0+1)\|F\|_\infty$, hence a Gr\"onwall bound for $\|F\|_\infty$.
Differentiating the PIDE in $\lambda$ and using Lemma~\ref{lem:J_eta_derivative} gives a closed integral inequality for
$\|F_\lambda\|_\infty$ of Gr\"onwall type. Finally $F_t$ is bounded from the equation.
Second-order bounds follow similarly under the additional assumptions.
\end{proof}

\section{MoE/Esscher: full weighted-TV bound and consistency}\label{app:MoE}
We use the weighted total-variation seminorm $\|\cdot\|_{TV,\delta}$ defined in
Section~\ref{sec:notations}, with $\delta\in(0,s_0)$. On a compact parameter set $K=\{k\in[k_-,k_+],\ b\in[b_-,b_+]\}$, $b_->\delta$,
define the exponential-moment factor $M_{k,b}(\delta):=\Big(\frac{b}{b-\delta}\Big)^k<\infty$.

\begin{lemma}[Gamma weighted-TV Lipschitz bound]\label{lem:gamma_wtv_lip}
For any $(k,b),(\hat k,\hat b)\in K$, the Gamma densities satisfy
\[
\int_0^\infty e^{\delta x}\big|f_{k,b}(x)-f_{\hat k,\hat b}(x)\big|\,dx
\ \le\ C_k\,|k-\hat k|+C_b\,|b-\hat b|,
\]
where $C_k,C_b<\infty$ depend only on $(k_\pm,b_\pm,\delta)$.
\end{lemma}

\begin{proof}[Proof]
Use the fundamental theorem of calculus along the segment between parameters and the identities
$\partial_b f_{k,b}=(k/b-x)f_{k,b}$ and $\partial_k f_{k,b}=(\ln b+\ln x-\psi(k))f_{k,b}$, where $\psi$ denotes the digamma function.
After weighting by $e^{\delta x}$, rewrite $e^{\delta x}f_{k,b}(x)=M_{k,b}(\delta)\, f_{k,b-\delta}(x)$
and bound the resulting expectations under $\Gamma(k,b-\delta)$ using standard moment formulas.
\end{proof}

\begin{prop}[MoE weighted-TV bound]\label{prop:moe_wtv_bound}
Let $\nu(\lambda,dx)=\sum_{m=1}^M \pi_m(\lambda)\Gamma(k_m,b_m)(dx)$ and
$\hat\nu(\lambda,dx)=\sum_{m=1}^M \hat\pi_m(\lambda)\Gamma(\hat k_m,\hat b_m)(dx)$, with all
$(k_m,b_m),(\hat k_m,\hat b_m)\in K$. Then for each $\lambda\in\Lambda$,
\[
\|\nu(\lambda,\cdot)-\hat\nu(\lambda,\cdot)\|_{TV,\delta}
\ \le\ \sum_{m=1}^M |\pi_m(\lambda)-\hat\pi_m(\lambda)|\,M_{k_m,b_m}(\delta)
\ +\ \sum_{m=1}^M \hat\pi_m(\lambda)\big(C_k|k_m-\hat k_m|+C_b|b_m-\hat b_m|\big).
\]
\end{prop}

\begin{proof}
Triangle inequality and Lemma~\ref{lem:gamma_wtv_lip} applied componentwise.
\end{proof}

We remark that for softmax gates $\pi(\lambda;\beta)=\mathrm{softmax}(\beta^\top B(\lambda))$ with bounded features $B$,
$\|\pi(\lambda;\beta)-\pi(\lambda;\beta^*)\|_1 \le L_\pi \|B\|_\infty\|\beta-\beta^*\|$ for a universal $L_\pi$.
Under standard compactness/identifiability assumptions, EM/MLE yields $\hat\theta\to\theta^*$ in probability
(e.g.\ \cite{McLachlanPeel2000,vanderVaart1998,White1982}).

\section{Boundary PCHIP slopes and short extrapolation}\label{app:PCHIP-boundary}

Let $\lambda_b\in\{\lambda_{\min},\lambda_{\max}\}$ and consider the neighbouring nodes
$\lambda_0=\lambda_b<\lambda_1<\lambda_2$ with steps $h_0=\lambda_1-\lambda_0$, $h_1=\lambda_2-\lambda_1$,
and $h:=\max(h_0,h_1)$. Let $d_0=(F_1-F_0)/h_0$, $d_1=(F_2-F_1)/h_1$ and let $m_b$ be the
Fritsch--Carlson boundary slope \cite{FritschCarlson1980}.

\begin{lemma}[First-order accuracy of the boundary slope]\label{lem:pchip_boundary_slope}
If $F\in C^2(\Lambda)$, then
\[
|m_b-F'(\lambda_b)|\ \le\ C\,h\,\|F''\|_{\infty,\Lambda},
\]
for a constant $C$ depending only on mesh regularity.
\end{lemma}

\begin{proof}[Proof sketch]
By MVT, $d_0=F'(\epsilon_0)$ and $d_1=F'(\epsilon_1)$ for some $\epsilon_0\in(\lambda_0,\lambda_1)$,
$\epsilon_1\in(\lambda_1,\lambda_2)$, hence $|d_0-F'(\lambda_0)|\le h_0\|F''\|_\infty$ and
$|d_1-F'(\lambda_0)|\le (h_0+h_1)\|F''\|_\infty$.
If $d_0d_1>0$, insert these bounds into the FC formula for $m_b$.
If $d_0d_1\le 0$, then $m_b=0$ and a zero of $F'$ exists in $(\lambda_0,\lambda_2)$, giving
$|F'(\lambda_0)|\le (h_0+h_1)\|F''\|_\infty$.
\end{proof}

\begin{corollary}[Short linear extrapolation error]\label{cor:short_extrap}
For $d>0$ small and $E[\widetilde F](\lambda_b\pm d):=\widetilde F(\lambda_b)\pm d\,m_b$,
\[
|E[\widetilde F](\lambda_b\pm d)-F(\lambda_b\pm d)|
\le d\,C h\,\|F''\|_{\infty,\Lambda}+\tfrac12 d^2\|F''\|_{\infty,\Lambda}.
\]

\noindent where $E[\widetilde F]$ denotes the short linear extrapolation of the interpolant introduced in (vii) of Theorem \ref{thm:consistency}
\end{corollary}

\newpage

\bibliographystyle{plain} 
\bibliography{SEbib.bib}     


\end{document}